\documentclass[reqno]{article}

\usepackage{amsthm}
\usepackage{amsmath}
\usepackage{amssymb}
\usepackage{graphicx}
\usepackage{color}
\usepackage{bbm}
\usepackage[top=2.5cm, bottom=2.5cm, left=2.5cm, right=2.5cm]{geometry}
\usepackage{enumerate}
\usepackage{faktor}
\usepackage{nicefrac}
\usepackage[nottoc]{tocbibind}
\usepackage{slashed}
\usepackage{authblk}
\usepackage[normalem]{ulem}

\usepackage[scr=euler]{mathalfa}

\def\mathclap#1{\text{\hbox to 0pt{\hss$\mathsurround=0pt#1$\hss}}}

\newcommand{\N}{\mathbb{N}}
\newcommand{\R}{\mathbb{R}}
\newcommand{\id}{\mathrm{id}}

\newcommand{\Ima}{\mathrm{Im}}
\newcommand{\Sp}{\mathbb{S}}

\newcommand{\ux}{\underline{x}}
\newcommand{\loc}{\mathrm{loc}}

\newcommand{\OG}{\overline{\Gamma}}
\newcommand{\ogk}{\overline{g}_K}
\newcommand{\ogS}{\overline{g}_{\mathbb{S}^d}}
\newcommand{\gr}{\mathrm{graph}}

\newcommand{\rd}{\partial}

\begin{document}

\numberwithin{equation}{section}
\newtheorem{theorem}[equation]{Theorem}
\newtheorem{remark}[equation]{Remark}
\newtheorem{assumption}[equation]{Assumption}
\newtheorem{claim}[equation]{Claim}
\newtheorem{lemma}[equation]{Lemma}
\newtheorem{definition}[equation]{Definition}
\newtheorem{corollary}[equation]{Corollary}
\newtheorem{proposition}[equation]{Proposition}
\newtheorem*{theorem*}{Theorem}
\newtheorem{conjecture}[equation]{Conjecture}
\newtheorem{example}[equation]{Example}

\setcounter{tocdepth}{3}

\title{The $C^0$-inextendibility of a class of FLRW spacetimes}
\author{Jan Sbierski\thanks{School of Mathematics, 
University of Edinburgh,
James Clerk Maxwell Building,
Peter Guthrie Tait Road, 
Edinburgh, 
EH9 3FD,
United Kingdom}}
\date{\today}

\maketitle

\begin{abstract}
This paper studies the singularity structure of FLRW spacetimes without particle horizons at the $C^0$-level of the metric. We  show that in the case of constant spatial curvature $K=+1$, and without any further assumptions on the scale factor, the big bang singularity is sufficiently strong to exclude continuous spacetime extensions to the past. On the other hand it is known that in the case of constant spatial curvature $K=-1$ continuous spacetime extensions through the big-bang  exist for certain choices of scale factor \cite{GalLin16}, giving rise to Milne-like cosmologies. Complementing these results we exhibit a geometric obstruction to continuous spacetime extensions  which is present for a large range of scale factors in the case $K=-1$.
\end{abstract}

\tableofcontents

\section{Introduction}

The study of low-regularity inextendibility properties of spacetime singularities is motivated by the attempt to understand and classify the strength of gravitational singularities on the one hand and, on the other hand, to rule out the possibility that the spacetime can be continued as a weak solution to the Einstein equations. 
We refer the reader also to the detailed discussion in \cite{Sbie22a}. The first continuous inextendibility results were obtained in \cite{Sbie15} for the Minkowski and Schwarzschild spacetimes. Further $C^0$-inextendibility results followed in \cite{GalLin16} for anti de-Sitter spacetime, for general timelike geodesically complete spacetimes in \cite{GalLinSbi17} and \cite{MinSuhr19}, while \cite{ChrusKli12} gave a conditional $C^0$-inextendibility result for expanding singularities. Here, in this paper, we continue the study of $C^0$-inextendibility results of Lorentzian manifolds by turning to another fundamental class of model spacetimes: the  FLRW spacetimes. This class of spacetimes describes homogeneous and isotropic cosmologies which feature a big bang singularity. The main aim of this work is to capture the different singularity structures of the big bang occurring in these models at the $C^0$-level of the metric.


\subsection{The class of FLRW spacetimes and the main theorems}

Let $K \in \{-1,0,1\}$ and consider the simply connected and complete $d$-dimensional Riemannian geometries $(\overline{M}_K, \overline{g}_K)$ of constant curvature $K$. We have $\overline{M}_K = \R^d$ for $K=-1,0$ and $\overline{M}_{+1} = \Sp^d$.\footnote{For $d=1$ this is pure convention. Let us also emphasise that we only consider these topologies for $\overline{M}_K$ in this paper. The choice of topology affects whether the spacetime is future one-connected, a key element in the proofs presented here (cf.\  Section \ref{SecDiscussion}).} The metric in polar normal coordinates is given by
\begin{equation} \label{EqMetricNormal}
\overline{g}_K = \begin{cases} dr^2 + \sinh^2(r) g_{\Sp^{d-1}} \quad &\textnormal{ for } K=-1 \\
dr^2 + r^2 g_{\Sp^{d-1}} \quad &\textnormal{ for } K=0 \\
dr^2 + \sin^2(r) g_{\Sp^{d-1}} \quad &\textnormal{ for } K=+1 
\end{cases}
\end{equation} 
around any point in $\overline{M}_K$, where $g_{\Sp^{d-1}}$ is the round metric on $\Sp^{d-1}$. In the cases of $K=-1,0$ the $r$-coordinate ranges from zero to infinity, while in the case $K=+1$, we have $r \in (0,\pi)$.

The class of $(d+1)$-dimensional FLRW spacetimes consists of Lorentzian manifolds $(M,g)$, where the manifold is of the form $M = (0, \infty) \times \overline{M}_K$ with a metric $g = -dt^2 + a(t)^2\overline{g}_K$. Here, $t$ is the canonical coordinate on the first factor of $M$ and $a: (0,\infty) \to (0,\infty)$ is the scale-factor which satisfies $a(t) \to 0$ for $t \to 0$ and is assumed to be at least twice continuously differentiable so that we have $g \in C^2$.
We define a time orientation by stipulating that $\rd_t$ is future directed. Clearly, FLRW spacetimes as defined above are globally hyperbolic. If the scale factor satisfies $\int\limits_0^1 \frac{1}{a(t)} \, dt = \infty$, we say that the FLRW spacetime does not possess particle horizons. If the integral is finite, it is said to possess particle horizons.

We now recall the definition of a spacetime extension. While the Lorentzian metrics may have limited regularity, the manifolds themselves are always assumed to be smooth, see also Remark 2.2 in \cite{Sbie22a}.
\begin{definition}
Let $(M,g)$ be a  Lorentzian manifold and let $\Gamma$ be a regularity class, for example $\Gamma = C^k$ with $k \in \N \cup \{\infty\}$ or $\Gamma = C^{0,1}_{\loc}$.  A \emph{$\Gamma$-extension of $(M,g)$} consists of a smooth isometric embedding $\iota : M \hookrightarrow \tilde{M}$ of $M$ into a Lorentzian manifold $(\tilde{M}, \tilde{g})$ of the same dimension as $M$ where $\tilde{g}$ is $\Gamma$-regular  and such that $\partial \iota(M) \subseteq \tilde{M}$ is non-empty.

If $(M,g)$ admits a $\Gamma$-extension, then we say that $(M,g)$ is \emph{$\Gamma$-extendible}, otherwise we say $(M,g)$ is \emph{$\Gamma$-inextendible}.
\end{definition}

\begin{definition}\label{DefFutureBoundary}
Let $(M,g)$ be a  time-oriented Lorentzian manifold and $\iota : M \hookrightarrow \tilde{M}$ a $C^0$-extension of $M$. The \emph{future boundary of $M$} is the set $\partial^+\iota(M) $ consisting of all points $\tilde{p} \in \tilde{M}$ such that there exists a smooth timelike curve $\tilde{\gamma} : [-1,0] \to \tilde{M}$ such that $\mathrm{Im}(\tilde{\gamma}|_{[-1,0)}) \subseteq \iota(M)$, $\tilde{\gamma}(0) = \tilde{p} \in \partial \iota(M)$, and $\iota^{-1} \circ \tilde{\gamma}|_{[-1,0)}$ is future directed in $M$.
\end{definition}
Clearly we have $\partial^+\iota(M) \subseteq \partial \iota(M)$. The past boundary $\partial^- \iota(M)$ is defined analogously. 

\begin{definition}
Let $(M,g)$  be a  time-oriented Lorentzian manifold and let $\Gamma$ be a regularity class that is equal to or stronger than $C^0$. A \emph{future $\Gamma$-extension of $(M,g)$} is a $\Gamma$-extension $\iota : M \hookrightarrow \tilde{M}$  of $M$ with $\partial^+\iota(M) \neq \emptyset$. If no such extension exists, then $(M,g)$ is said to be \emph{future $\Gamma$-inextendible}.
\end{definition}
Past $\Gamma$-extensions are defined analogously. We can now state the main results.

\begin{theorem}\label{Thm1}
Let $d \geq 2$. Consider the class of $(d+1)$-dimensional FLRW spacetimes without particle horizons and with $K=+1$ as defined above. Each such spacetime $(M,g)$ is past $C^0$-inextendible.
\end{theorem}

\begin{theorem}\label{Thm2}
Let $d \geq 2$. Consider the class of $(d+1)$-dimensional FLRW spacetimes without particle horizons and with $K=-1$ as defined above. Assume in addition that the scale factor satisfies $a(t) \cdot e^{\int_t^1\frac{1}{a(t')} \, dt'} \to \infty$ for $t \to 0$. Then each such spacetime $(M,g)$ is past $C^0$-inextendible.
\end{theorem}

\subsection{Previous results and remarks concerning the main theorems}

It is well-known that by computing scalar curvature quantities one can establish the past $C^2$-inextendibility of a large range of FLRW spacetimes. However, not much is known about the low-regularity\footnote{By low-regularity we mean in this paper any regularity between $C^0$ and $C^1$.} inextendibility properties of this class of spacetimes; the only available result establishes the past $C^{0,1}_{\mathrm{loc}}$-inextendibility of FLRW spacetimes \emph{with particle horizons} for $d \geq 1$, \cite{Sbie22a}. 

On the other hand there are a couple of \emph{$C^0$-extendibility} results available. \textbf{Firstly}, in dimension $d=1$ all FLRW spacetimes are past $C^0$-extendible, which also shows that the assumption $d \geq 2$ is crucial in our main theorems. To see this, let $M = (0, \infty) \times \R$ with Lorentzian metric $g = -dt^2 + a(t)^2 \, dx^2$. Assume $a(t) \to 0$ for $t \to 0$. We introduce the null coordinates $v := x - \int\limits_t^1 \frac{1}{a(t')} \, dt'$ and $u:= x + \int\limits_t^1 \frac{1}{a(t')} \, dt'$ and also a new time-coordinate $\tau := \int\limits_0^t a(t') \, dt'$. In $(\tau,v)$ coordinates the metric $g$ takes the form $$g = a\big(t(\tau)\big)^2 \,dv^2 - [dv \otimes d\tau + d \tau \otimes dv]\;,$$
while in $(\tau,u)$ coordinates the metric $g$ takes the form $$g = a\big(t(\tau)\big)^2 \, du^2 + [du \otimes d \tau + d \tau \otimes du] \;.$$
The differentiable structure of each of these coordinates furnishes a continuous extension beyond $\tau = 0$.\footnote{This has been observed in Section 3.2.\ of \cite{GalLin16}.} In the case of particle horizons, these two continuous extensions terminate the same past directed causal curves, see the discussion in Section 6 of \cite{Sbie22b}. In the case of no particle horizons, they terminate a disjoint set of past directed causal curves. We can also replace the second factor in $M$ by $\Sp^1$ and then define the null coordinates $u$ and $v$ modulo $2\pi$. This still gives us two $C^0$-extensions.

\textbf{Secondly}, Galloway and Ling \cite{GalLin16} showed for $d \geq 2$ that if $(M,g)$ is an FLRW spacetime with $K=-1$ and without particle horizons, where the scale factor satisfies in addition 
\begin{equation} \label{EqGalLing}
\lim_{t \to 0} \dot{a}(t) = 1 \quad \textnormal{ and } \quad \lim_{t \to 0} a(t) \exp(\int_t^1 \frac{1}{a(t')} \, dt') \in (0, \infty) \;,
\end{equation}
then $(M,g)$ is past $C^0$-extendible. The authors called such spacetimes \emph{Milne-like spacetimes}. Scale factors of the form $a(t) = t + t^{1 + \varepsilon}$, with $\varepsilon>0$, for example satisfy the assumptions \eqref{EqGalLing}, \cite{Ling20}. On the other hand, scale factors of the form $a(t) = t^{1 + \varepsilon}$ satisfy the assumptions of Theorem \ref{Thm2}. Note that while Theorem \ref{Thm2} complements the $C^0$-extendibility criterion of Galloway and Ling, it leaves open the regime of scale factors for which $a(t) \exp(\int_t^1 \frac{1}{a(t')} \, dt')$ is bounded but does not converge for $t \to 0$.

So far we have only discussed the past-extendibility properties of FLRW spacetimes through the big-bang singularity. It is shown in Theorem 3.5 of \cite{Sbie22a} that for  scale factors satisfying $\int\limits_1^\infty \frac{a(t)}{\sqrt{a(t)^2 +1}} \, dt = \infty$ the FLRW spacetime is future timelike geodesically complete, which implies the future $C^0$-inextendibility by \cite{GalLinSbi17}.\footnote{See also Corollary 2.10 in \cite{GalLin16} for a related result.} Note that if a Lorentzian manifold is future and past $C^0$-inextendible, then it is (globally) $C^0$-inextendible. This follows from the following lemma, which is a reformulation of Lemma 2.17 in \cite{Sbie15}:
\begin{lemma} \label{LemFuturePastExt}
Let $(M,g)$ be a  time-oriented Lorentzian manifold and $\iota : M \hookrightarrow \tilde{M}$ a $C^0$-extension of $M$. Then $\partial^+\iota(M) \cup \partial^- \iota(M) \neq \emptyset$.
\end{lemma}
For this reason we focus solely on past-inextendibility in this paper.

Related to the low-regularity inextendibility results for cosmological spacetimes considered in this paper are the generalisations of Hawking's singularity theorem to Lorentzian manifolds which have a merely $C^1$-regular metric \cite{Gr20}, \cite{KuOhScSt22}.
Synthetic versions of Hawking's singularity theorem are also derived in \cite{AlGrKuSa19} in the setting of Lorentzian length spaces and in \cite{CaMo20} for Lorentzian metric measure spaces. An inextendibility result due to timelike geodesic completeness in the setting of Lorentzian length spaces was obtained in \cite{GraKuSa19}.

\subsection{Discussion of proofs and outline of paper} \label{SecDiscussion}

The proofs of the main theorems are by contradiction and develop further methods originating in the works \cite{Sbie15}, \cite{Sbie18} by the author. Assuming the existence of a past $C^0$-extension, one uses the global hyperbolicity of the FLRW spacetimes to ensure that a past boundary chart exists in which part of the past boundary is given as a Lipschitz graph. This statement, proved in \cite{Sbie18}, is recalled in Proposition \ref{PropBoundaryChart} in Section \ref{SecFundResults}. We also establish the future one-connectedness of the FLRW spacetimes in Proposition \ref{PropFutureOne} in Section \ref{SecFLRW}  (in the case of $K=+1$, it is here that the restriction $d \geq 2$ enters) which is used to relate timelike diamonds in the past boundary chart to those in the original FLRW spacetime. The proof of Proposition \ref{PropFutureOne} follows ideas that were originally developed in \cite{Sbie15}. The main geometric obstruction to the past $C^0$-extendibility in the case of $K=+1$ is captured in Proposition \ref{PropK1} in Section \ref{SecFLRW1}. Crucially using the absence of particle horizons and the compactness of the Cauchy hypersurfaces, we show that a) the timelike future of a past inextendible timelike curve equals the whole spacetime, and b) given any point in the FLRW spacetime, its timelike past contains a whole Cauchy hypersurface of constant time $t$ for $t$ small enough. It is then shown in Section \ref{SecThm1}, using a new causality theoretic argument, that this is indeed a contradiction to the geometric bounds provided by the past boundary chart, which concludes the proof of Theorem \ref{Thm1}. 

The proof of Theorem \ref{Thm2} is more involved. Proposition \ref{PropTIFs} in Section \ref{SecFLRW2} can be seen as parametrising the past causal boundary (\cite{GeKrPe72}) of FLRW spacetimes $(M,g)$ with $K=-1$ and without particle horizons in terms of a null $u$-coordinate and angular coordinates on $\Sp^{d-1}$. It is also shown that given two terminal indecomposable future sets (TIFs) of which neither is contained in the other, then those must be given by finite $u$-coordinates and \emph{different} angular coordinates. Given past inextendible timelike curves $\gamma_1$ and $\gamma_2$, parameterised by the time coordinate $t$ and generating such TIFs, it is shown in Proposition \ref{PropBlowUp} that 
\begin{equation}\label{EqMO}
d_{\Sigma_{t_0} \setminus J^-(r,M)} \big(\gamma_1(t_0), \gamma_2(t_0)\big) \to \infty \quad \textnormal{ for } t_0 \to 0\;, 
\end{equation}
where $\Sigma_{t_0} = \{t = t_0\}$ and $r$ is a point such that $\gamma_1$ and $\gamma_2$ lie in the complement of $J^-(r,M)$. The statement \eqref{EqMO} can be viewed as the main geometric obstruction to past $C^0$-extensions; it can be seen as a geometric consequence of the additional assumption on the scale factor in Theorem \ref{Thm2}. Note that the exclusion of $J^-(r,M)$ in \eqref{EqMO} is essential -- otherwise the distance in $\Sigma_{t_0}$ of the two points $\gamma_1(t_0)$ and $\gamma_2(t_0)$ would go to zero for $t_0 \to 0$. 

In order to lead \eqref{EqMO} together with the assumption of a past $C^0$-extension to a contradiction, we need to refine the past boundary charts from Proposition \ref{PropBoundaryChart}:  we show in Proposition \ref{PropSpacelikeB} that we can find a point at which the past boundary  is differentiable and, using $d \geq 2$, admits a spacelike tangent vector.\footnote{The appendix \ref{SecApp} recalls a few fundamental properties of tangent cones and Lipschitz hypersurfaces needed for this construction.} Starting from this set-up, the proof of Theorem \ref{Thm2} in Section \ref{SecThm2} then constructs past inextendible timelike curves $\gamma_1$ and $\gamma_2$ in this past boundary chart such that neither TIF is contained in the other (here, the future one-connectedness is used again). Moreover, a point $r$ in the past boundary chart is constructed such that $J^-(r,M)$ does not contain $\gamma_1$ nor $\gamma_2$ -- placing us in the situation of \eqref{EqMO}. For this construction it is again essential to relate the causality relations in the past boundary chart with those in the FLRW spacetime. This is done by Proposition \ref{PropIntNG} in Section \ref{SecFundResults} and Proposition \ref{PropNGI} in Section \ref{SecFLRW2}. Using a similar argument as in \cite{Sbie15}, one now constructs a curve in the past boundary chart which connects the points $\gamma_1(t_0)$ and $\gamma_2(t_0)$  in $\Sigma_{t_0} \setminus J^-(r,M)$ and the length of which is uniformly bounded in $t_0$. This derives the contradiction and concludes the proof of Theorem \ref{Thm2}.

Let us point out that no attempt has been made in this paper to state the most general results possible, but the emphasis lies on the introduction of new methods to the study of low-regularity inextendibility problems. For example, the proof of Theorem \ref{Thm1} really only requires the future one-connectedness of the spacetime, the absence of particle horizons, and the compactness of the spacelike geometries.

\subsection*{Acknowledgements}

I would like to acknowledge support through the Royal Society University Research Fellowship URF\textbackslash R1\textbackslash 211216. I am also grateful to two anonymous referees for helpful comments on a previous version of this manuscript.

\section{Properties of general $C^0$-extensions} \label{SecFundResults}

We begin with laying out our conventions regarding Lorentzian causality theory that we use in this paper for Lorentzian manifolds $(M,g)$ with a continuous metric. A \emph{timelike curve} is a piecewise smooth curve which has a timelike tangent everywhere -- and at the points of discontinuity of the tangent the right and left tangent vectors lie in the same connectedness component of the timelike double cone of tangent vectors. Similarly we define a \emph{causal curve} as a piecewise $C^1$ curve which has a causal, non-vanishing tangent everywhere -- and at points of discontinuity of the tangent the right and left tangent vectors lie in the same connectedness component of the causal double cone of tangent vectors with the origin removed. Let $(M,g)$ be in addition time-oriented. For $p \in M$ we denote the \emph{timelike future} of $p$ in $M$ by  $I^+(p,M)$, which is  the set of all points $q \in M$ such that there is a future directed timelike curve  from $p$ to $q$. The \emph{causal future} of $p$ in $M$, denoted by $J^+(p,M)$, is the set which contains $p$ and all points $q \in M$ such that there is a future directed causal curve from $p$ to $q$. The sets $I^-(p,M)$ and $J^-(p,M)$ are defined analogously.

 Note that the past and future boundary, as defined in Definition \ref{DefFutureBoundary},  interchange under a change of time orientation of $(M,g)$. It is thus sufficient to focus in the following on the future boundary.
We recall the following fundamental result which is proved in \cite{Sbie18}, Proposition 2.2.

\begin{proposition}\label{PropBoundaryChart}
Let $\iota : M \hookrightarrow \tilde{M}$ be a $C^0$-extension of a time-oriented globally hyperbolic Lorentzian manifold $(M,g)$, $g \in C^2$, with Cauchy hypersurface $\Sigma$  and let $\tilde{p} \in \partial^+ \iota(M)$. For every $\delta >0$ there exists a chart $\tilde{\varphi} : \tilde{U} \to(-\varepsilon_0, \varepsilon_0) \times  (-\varepsilon_1, \varepsilon_1)^{d} $, $\varepsilon_0, \varepsilon_1 >0$ with the following properties
\begin{enumerate}[i)]
\item $\tilde{p} \in \tilde{U}$ and $\tilde{\varphi}(p) = (0, \ldots, 0)$
\item $|\tilde{g}_{\mu \nu} - m_{\mu \nu}| < \delta$, where $m_{\mu \nu} = \mathrm{diag}(-1, 1, \ldots , 1)$
\item There exists a Lipschitz continuous function $f : (-\varepsilon_1, \varepsilon_1)^d \to (-\varepsilon_0, \varepsilon_0)$ with the following property: 
\begin{equation}\label{PropF1}
\{(x_0,\underline{x}) \in (-\varepsilon_0, \varepsilon_0) \times (-\varepsilon_1, \varepsilon_1)^{d} \; | \: x_0 < f(\underline{x})\} \subseteq \tilde{\varphi} \big( \iota\big(I^+(\Sigma,M)\big)\cap \tilde{U}\big)
\end{equation} and 
\begin{equation}\label{PropF2}
\{(x_0,\underline{x}) \in (-\varepsilon_0, \varepsilon_0) \times (-\varepsilon_1, \varepsilon_1)^{d}  \; | \: x_0 = f(\underline{x})\} \subseteq \tilde{\varphi}\big(\partial^+\iota(M)\cap \tilde{U}\big) \;.
\end{equation}
Moreover, the set on the left hand side of \eqref{PropF2}, i.e. the graph of $f$, is achronal\footnote{With respect to \emph{smooth} timelike curves.} in $(-\varepsilon_0, \varepsilon_0) \times  (-\varepsilon_1, \varepsilon_1)^{d}$.
\end{enumerate}
In particular $\delta>0$ can be chosen so small such that $\tilde{g}_n := -\frac{1}{2} dx_0^2 + dx_1^2 + \ldots + dx_d^2 \prec \tilde{g} \prec -2 dx_0^2 + dx_1^2 + \ldots + dx_d^2 =:\tilde{g}_w $.\footnote{The relation $g_1 \prec g_2$ for two Lorentzian metrics means that if $X$ is a timelike vector with respect to $g_1$, then it is also timelike with respect to $g_2$.} The time orientation on $\tilde{U}$ is fixed by stipulating that $\rd_0$ is future directed. This agrees with the induced time orientation of $M$ below the graph of $f$.
\end{proposition}

We also introduce the abbreviations $\tilde{U}_{<} := \{x_0 < f(\underline{x})\}$ and $\tilde{U}_{\leq} := \{x_0 \leq f(\underline{x})\}$. 
A chart as in Proposition \ref{PropBoundaryChart} is called a \emph{future boundary chart} around $\tilde{p}$.
Note that 
\begin{equation} \label{EqAddPropBC}
\textnormal{any past directed causal curve in $\tilde{U}$ starting below the graph of $f$ remains below the graph of $f$.}
\end{equation}
Otherwise, if it were to cross the graph of $f$, it would give rise, via $\iota^{-1}$,  to a past directed past inextendible causal curve in $M$ which starts in $I^+(\Sigma,M)$ but does not intersect $\Sigma$ -- which contradicts $\Sigma$ being a Cauchy hypersurface.

It also follows directly from the inclusion relation of the light cones that in the setting of Proposition \ref{PropBoundaryChart} we have $I^+_{\tilde{g}_n}(\tilde{q}, \tilde{U}) \subseteq I^+_{\tilde{g}}(\tilde{q}, \tilde{U}) \subseteq I^+_{\tilde{g}_w}(\tilde{q}, \tilde{U})$ for any $\tilde{q} \in \tilde{U}$ -- and similarly for the timelike past and causal future/past.

The next proposition refines Proposition \ref{PropBoundaryChart} by showing that in at least three spacetime dimensions one can find a future boundary point such that the graph of $f$ has a spacelike direction at this point. Note that the future boundary can of course be null everywhere which shows that the result cannot hold in two spacetime dimensions. This refinement of Proposition \ref{PropBoundaryChart} is needed for the proof of Theorem \ref{Thm2}.

\begin{proposition} \label{PropSpacelikeB}
Let $(M,g)$ be a  $(d+1)$-dimensional time oriented and globally hyperbolic Lorentzian manifold $(M,g)$ with $g \in C^2$, $\iota : M \hookrightarrow \tilde{M}$ a $C^0$-extension, and $\rd^+ \iota(M) \neq \emptyset$. Let $d \geq 2$. Then there exist a $\tilde{p} \in \rd^+ \iota (M)$, a future boundary chart $\tilde{\varphi} : \tilde{U} \to(-\varepsilon_0, \varepsilon_0) \times  (-\varepsilon_1, \varepsilon_1)^{d}$ around $\tilde{p}$ as in Proposition \ref{PropBoundaryChart} such that $f : (-\varepsilon_1, \varepsilon_1)^d \to (-\varepsilon_0, \varepsilon_0)$ is differentiable at $0$ with $df(0) (\rd_1) = 0$.
\end{proposition}

\begin{proof}
By assumption, there is a $\tilde{q} \in \rd^+ \iota(M)$. By Proposition \ref{PropBoundaryChart} we can find a future boundary chart $\tilde{\psi} : \tilde{V} \to (-\delta_0, \delta_0) \times (-\delta_1, \delta_1)^d$ with Lipschitz continuous graphing function $h : (-\delta_1, \delta_1)^d \to (-\delta_0, \delta_0)$. By Rademacher's theorem, $h$ is differentiable almost everywhere. Let $\underline{y}_o \in (-\delta_1, \delta_1)^d$ be a point where $dh$ exists. We set $H(y) :=y_0 - h(\underline{y})$, define $\tilde{p} := \tilde{\psi}^{-1} \big(h(\underline{y}_o), \underline{y}_o\big)$ and $\tilde{B} := \tilde{\psi}^{-1}(\mathrm{graph} (h))$. Then $H$ is differentiable at $\tilde{p}$ and we have $\mathrm{ker} (dH|_{\tilde{p}}) = T_{\tilde{p}}\tilde{B}$.\footnote{Indeed, one can show that the global hyperbolicity of $M$ implies that $dH|_{\tilde{p}}$ has to be null or timelike, but this is not needed for the proof.} Since we have assumed $d \geq 2$ there exists a spacelike unit vector $e_1 \in T_{\tilde{p}} \tilde{V}$ with $dH(e_1) = 0$.  We complement $e_1$ to an ONB $e_\mu$ at $\tilde{p}$ such that $e_0$ is future directed timelike. By an affine change of coordinates we find a smooth chart $\tilde{\varphi} : \tilde{V} \supseteq \tilde{U} \to(-\varepsilon_0, \varepsilon_0) \times  (-\varepsilon_1, \varepsilon_1)^{d}$ which is centred at $\tilde{p}$ and such that we have $\frac{\rd}{\rd x^\mu} = e_\mu$ at $\tilde{p}$. After making $\varepsilon_0>0$ smaller if necessary it follows that the smooth curve $ (-\varepsilon_0 ,0] \ni s \overset{\tilde{\gamma}}{\mapsto} \tilde{\varphi}^{-1}(s,\underline{0})$ is timelike. Moreover, with respect to the first chart $\tilde{\psi}$, it maps into $\tilde{V}_<$ if the right endpoint is deleted\footnote{Suppose it did not and we could find an $s_0 \in (-\varepsilon_0,0)$ with $\tilde{\gamma}(s_0) \in \{y_0 \geq h(\underline{y})\}$. Concatenating  $\tilde{\gamma}|_{[s_0,0]}$ to the left with a curve just moving backwards in $y_0$ gives then a piecewise smooth timelike curve connecting points on the graph of $h$, in contradiction to its achronality.} -- and thus in particular into $\iota(M)$. We now follow line by line the proof of Proposition 2.2 in \cite{Sbie18} which shows that after making $\varepsilon_0, \varepsilon_1 >0$ smaller if necessary the chart $\tilde{\varphi}$ becomes a future boundary chart with graphing function $f$ as in Proposition \ref{PropBoundaryChart}. 

 Note that  the curve $\tilde{\gamma}$ lies below the respective graphs in both charts and thus, by continuity, $\tilde{\psi} \circ \tilde{\varphi}^{-1} $  maps $\tilde{U}_<$ into $\tilde{V}_<$. Since every point on $\mathrm{graph}(f)$ is the limit point of a timelike curve in $\tilde{U}_<$ whose image in $\tilde{V}_<$ can only have a limit point on $\mathrm{graph}(h)$, it follows that also  $\tilde{U}_{\leq}$ is mapped into $\tilde{V}_{\leq}$. Finally, since $\tilde{\psi} \circ \tilde{\varphi}^{-1} $ is a diffeomorphism, it follows that $\tilde{\varphi}(\tilde{B} \cap \tilde{U}) = \mathrm{graph} (f)$. 
Lemma \ref{LemLip} implies that $f$ is differentiable at $0$ and we have $T_{\tilde{p}}\tilde{B} \ni e_1 = \frac{\rd}{\rd y_1} = d(f, \mathrm{id}_{\R^d} )|_0 \big( \frac{\rd}{\rd y_1}\big)$, which implies $df|_0(\rd_1) = 0$. (Here, $(f,\mathrm{id}_{\R^d})(\underline{x}) = (f(\underline{x}), \underline{x})$.) This concludes the proof.
\end{proof}

The next proposition gives a criterion which rules out that for $\tilde{r}, \tilde{s} \in \tilde{U}_<$ with $\tilde{r}$ in the future of $\tilde{s}$, a future directed timelike curve in $M$, starting at $\iota^{-1}(\tilde{r})$, leaves the chart $\tilde{U}$ and then enters again the future of $\tilde{s}$ in $\tilde{U}_<$ by passing through the light cone of $\tilde{s}$ in $\tilde{U}_<$.
\begin{proposition} \label{PropIntNG}
Let $(M,g)$ be a time-oriented globally hyperbolic Lorentzian manifold with $g \in C^2$. Let $\iota : M \hookrightarrow \tilde{M}$ be a $C^0$-extension, $\tilde{p} \in \rd^+ \iota(M)$, and $\tilde{\varphi} : \tilde{U} \to(-\varepsilon_0, \varepsilon_0) \times  (-\varepsilon_1, \varepsilon_1)^{d}$ a future boundary chart around $\tilde{p}$. Furthermore, let $\tilde{s}, \tilde{r}\in \tilde{U}_<$ with $\tilde{r} \in I^+(\tilde{s}, \tilde{U})$ and let $\tilde{t} \in I^+(\tilde{s}, \tilde{U})$ be such that $J^+_{\tilde{g}_w}(\tilde{s}, \tilde{U}) \cap J^-_{\tilde{g}_w}(\tilde{t}, \tilde{U})$ is compact in $\tilde{U}$. Assume that $s := \iota^{-1}(\tilde{s})$, $r:= \iota^{-1}(\tilde{r}) \in M$ are such that none of the future directed null geodesics in $M$ emanating from $s$ intersect any of those emanating from $r$. Then 
\begin{equation*}
\iota\big(J^+(r,M)\big) \cap  \Big( \big[ J^+(\tilde{s}, \tilde{U}_<) \cap J^-_{\tilde{g}_w}(\tilde{t}, \tilde{U})  \big]\setminus J^+(\tilde{r}, \tilde{U}_<)\Big) = \emptyset \;.
\end{equation*}
\end{proposition}

\begin{proof}
Assume the conclusion does not hold. Then there exists a point $\tilde{z} \in  \big[ J^+(\tilde{s}, \tilde{U}_<) \cap J^-_{\tilde{g}_w}(\tilde{t}, \tilde{U}) \big] \setminus J^+(\tilde{r}, \tilde{U}_<)$ and a future directed causal curve $\gamma : [0,1] \to M$ from $r$ to $z:= \iota^{-1}(\tilde{z})$. Clearly, $\iota \circ \gamma =: \tilde{\gamma}$ cannot be contained in $\tilde{U}_<$ (and not in $\tilde{U}$, since in this case it would need to cross $\rd \iota(M)$) but has to leave $\tilde{U}$ and enter $\tilde{U}_<$ again. Consider the maximal domain of the form $(\tau_0,1]$ such that $\tilde{\gamma}\big((\tau_0,1]\big) \subseteq \tilde{U}$. Since $\tilde{\gamma}(1) = \tilde{z} \in \tilde{U}_<$, $\tilde{\gamma}_{(\tau_0, 1]}$ maps into $\tilde{U}_<$. Since the time orientations on $M$ and $\tilde{U}_<$ agree, it is future directed in $\tilde{U}_<$. Thus in particular $\tilde{\gamma}\big((\tau_0,1]\big) \subseteq J^-_{\tilde{g}_w}(\tilde{z}, \tilde{U}) \subseteq J^-_{\tilde{g}_w}(\tilde{t}, \tilde{U})$. By the  compactness of $J^+_{\tilde{g}_w}(\tilde{s}, \tilde{U}) \cap J^-_{\tilde{g}_w}(\tilde{t}, \tilde{U})$ the curve $\tilde{\gamma}_{(\tau_0,1]}$ has to intersect $\rd J^+(\tilde{s}, \tilde{U}_<)$, say for $\tau_0 < \tau_1 \leq 1$. Since $\tilde{g}|_{\tilde{U}_<} \in C^2$, there exists a future directed null geodesic $\tilde{\sigma} : [0,1] \to \tilde{U}_<$ from $\tilde{s}$ to $\tilde{\gamma}(\tau_1)$. 
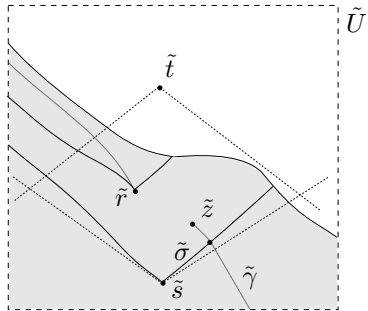
\begin{figure}[h]
\centering
 \def\svgwidth{5cm}
\begingroup%
  \makeatletter%
  \providecommand\color[2][]{%
    \errmessage{(Inkscape) Color is used for the text in Inkscape, but the package 'color.sty' is not loaded}%
    \renewcommand\color[2][]{}%
  }%
  \providecommand\transparent[1]{%
    \errmessage{(Inkscape) Transparency is used (non-zero) for the text in Inkscape, but the package 'transparent.sty' is not loaded}%
    \renewcommand\transparent[1]{}%
  }%
  \providecommand\rotatebox[2]{#2}%
  \newcommand*\fsize{\dimexpr\f@size pt\relax}%
  \newcommand*\lineheight[1]{\fontsize{\fsize}{#1\fsize}\selectfont}%
  \ifx\svgwidth\undefined%
    \setlength{\unitlength}{316.89460094bp}%
    \ifx\svgscale\undefined%
      \relax%
    \else%
      \setlength{\unitlength}{\unitlength * \real{\svgscale}}%
    \fi%
  \else%
    \setlength{\unitlength}{\svgwidth}%
  \fi%
  \global\let\svgwidth\undefined%
  \global\let\svgscale\undefined%
  \makeatother%
  \begin{picture}(1,0.81208411)%
    \lineheight{1}%
    \setlength\tabcolsep{0pt}%
    \put(0,0){\includegraphics[width=\unitlength,page=1]{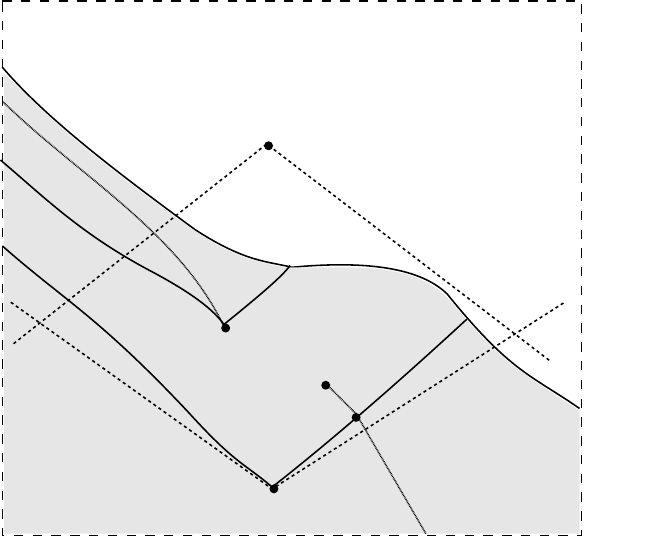}}%
    \put(0.90214959,0.74015832){\color[rgb]{0,0,0}\makebox(0,0)[lt]{\lineheight{1.25}\smash{\begin{tabular}[t]{l}$\tilde{U}$\end{tabular}}}}%
    \put(0.28769216,0.27728499){\color[rgb]{0,0,0}\makebox(0,0)[lt]{\lineheight{1.25}\smash{\begin{tabular}[t]{l}$\tilde{r}$\end{tabular}}}}%
    \put(0.51506842,0.24750948){\color[rgb]{0,0,0}\makebox(0,0)[lt]{\lineheight{1.25}\smash{\begin{tabular}[t]{l}$\tilde{z}$\end{tabular}}}}%
    \put(0.62604973,0.06885669){\color[rgb]{0,0,0}\makebox(0,0)[lt]{\lineheight{1.25}\smash{\begin{tabular}[t]{l}$\tilde{\gamma}$\end{tabular}}}}%
    \put(0.43656951,0.03096064){\color[rgb]{0,0,0}\makebox(0,0)[lt]{\lineheight{1.25}\smash{\begin{tabular}[t]{l}$\tilde{s}$\end{tabular}}}}%
    \put(0.42032831,0.62105627){\color[rgb]{0,0,0}\makebox(0,0)[lt]{\lineheight{1.25}\smash{\begin{tabular}[t]{l}$\tilde{t}$\end{tabular}}}}%
    \put(0.44468997,0.13382135){\color[rgb]{0,0,0}\makebox(0,0)[lt]{\lineheight{1.25}\smash{\begin{tabular}[t]{l}$\tilde{\sigma}$\end{tabular}}}}%
  \end{picture}%
\endgroup%

      \caption{The proof of Proposition \ref{PropIntNG}.} \label{FigProofNG}
\end{figure}

Since $M$ is acausal, we have $s \notin J^+(r,M)$. Thus $\tilde{\sigma}$ connects a point which is not in $\iota\big(J^+(r,M)\big)$ with a point in $\iota \big(J^+(r,M)\big)$. Hence, there is $0<s_0 \leq 1$ such that $\tilde{\sigma}(s_0) \in \iota \big(\rd J^+(r,M)\big)$. But then there exists a future directed null geodesic from $r$ to $\iota^{-1}\big(\tilde{\sigma}(s_0)\big)$, which is a contradiction.
\end{proof}

\begin{definition}
Let $(M,g)$ be a time-oriented Lorentzian manifold with a continuous metric. 
\begin{enumerate}
\item Two future directed timelike curves $\gamma_i : [0,1] \to M$, $i = 0,1$, with $\gamma_0(0) = \gamma_1(0)$ and $\gamma_0(1) = \gamma_1(1)$ are called \emph{timelike homotopic with fixed endpoints} if, and only if, there exists a continuous map $\Gamma : [0,1] \times [0,1] \to M$ such that $\Gamma(t, \cdot)$ is a future directed timelike curve from $\gamma_0(0)$ to $\gamma_0(1)$ for all $t \in [0,1]$  and, moreover, $\Gamma(0, \cdot) = \gamma_0(\cdot)$ and $\Gamma(1,\cdot) = \gamma_1(\cdot)$. The map $\Gamma$ is also called a \emph{timelike homotopy with fixed endpoints between $\gamma_0$ and $\gamma_1$}.
\item We say that $(M,g)$ is \emph{future one-connected} if, and only if, for all $p, q \in M$, any two future directed timelike curves from $p$ to $q$ are timelike homotopic with fixed endpoints.
\end{enumerate}
\end{definition}

The next lemma is found in \cite{Sbie22a}, Lemma 2.12.
\begin{lemma}\label{LemCausalHomotopy}
Let $(M,g)$ be a time-oriented Lorentzian manifold with $g \in C^0$ and let $\iota : M \hookrightarrow \tilde{M}$ be a $C^0$-extension of $M$. Moreover, let $\gamma : [0,1] \to M$ be a future directed timelike curve and let $\tilde{U} \subseteq \tilde{M}$ be an open set. Assume that $\tilde{\gamma} := \iota \circ \gamma$ maps into $\tilde{U}$ and that $J^+(\tilde{\gamma}(0), \tilde{U}) \cap J^-(\tilde{\gamma}(1), \tilde{U}) \subset \subset \tilde{U}$ is compactly contained in $\tilde{U}$.

Let $\Gamma : [0,1] \times [0,1] \to M$ be a causal homotopy of $\gamma$ with fixed endpoints, i.e., \begin{enumerate}
\item $s \mapsto \Gamma(r ; s)$ is a future directed causal curve for all $r \in [0,1]$ with $\Gamma(r;0) = \gamma(0)$ and $\Gamma(r;1) = \gamma(1)$
\item $\Gamma(0;s) = \gamma(s)$ for all $s \in [0,1]$.
\end{enumerate}

Then $\iota \circ \Gamma$ maps into $\tilde{U}$.
\end{lemma}

\section{Properties of FLRW spacetimes without particle horizons} \label{SecFLRW}

\begin{proposition} \label{PropFutureOne}
Let $d \geq 2$. Any $(d+1)$-dimensional FLRW spacetime as defined in the introduction is future one-connected.
\end{proposition}

Although it is not needed in this paper, let us emphasise that this proposition also applies to FLRW spacetimes with particle horizons.

\begin{proof}
The proof follows along the same lines as that of Proposition 4.12 in \cite{Sbie15}.
We first note that future one-connectedness is a statement about the conformal class of the spacetime under consideration. After setting $\tau(t) := \int\limits_1^t \frac{1}{a(t')} \, dt'$, the FLRW metric becomes $g = a(t)^2 (-d\tau^2 + \overline{g}_K)$. It thus suffices to show that  the Lorentzian manifolds $(N,-d\tau^2 +\overline{g}_K)$ are future one-connected, where $N = (b,c) \times \overline{M}_K$ and $-\infty \leq b < c \leq \infty$, depending on the scale factor $a$. 

Let $\gamma : [0,T] \to N$ be a piecewise smooth timelike curve. Without loss of generality we can assume that $\gamma$ is parameterised by the $\tau$-coordinate, i.e., 
\begin{equation}
\label{EqLift}
\gamma(s) = (s, \overline{\gamma}(s))\;.
\end{equation} 
Since $\gamma$ is timelike, it follows that 
\begin{equation}
\label{EqBoundSpC}
||\dot{\overline{\gamma}}(s)||_{\overline{g}_K} <1 \quad \textnormal{ for all s }\in [0,T]\;.
\end{equation} Vice versa, any such spatial curve $\overline{\gamma} : [0,T] \to \overline{M}_K$ which satisfies \eqref{EqBoundSpC} lifts to a timelike curve in $N$ via \eqref{EqLift}. 

Consider a continuous map $\OG : [0,1] \times [0,T] \to \overline{M}_K$ such that for each $u \in [0,1]$ $s \mapsto \OG(u,s)$ is a piecewise smooth curve and such that $\OG(u,0)$ and $\OG(u,1)$ are independent of $u \in [0,1]$. We call such a map a homotopy of piecewise smooth curves with fixed endpoints. If it satisfies in addition $||\partial_s \OG(u,s)||_{\overline{g}_K} <1$ for all $s \in [0,T]$ and $u \in [0,1]$, then it also lifts to a timelike homotopy with fixed endpoints between the timelike curves $s \mapsto (s, \OG(0,s))$ and $s \mapsto (s, \OG(1,s))$. We say in this case that the homotopy $\OG$ has the \emph{timelike lifting property}.

We consider first the simpler cases $K=-1,0$.\footnote{Indeed, in those cases one could just quote Proposition 2.7 in \cite{GalLin16}.} Let $\gamma_i$, $i = 1,2$ be two future directed timelike curves in $N$ which have the same endpoints. Since the metric is independent of $t$, without loss of generality we can assume that the curves are defined on the interval $[0,T]$ and are parameterised by the $t$-coordinate. Consider the projections $\overline{\gamma}_i$ of these curves onto $\overline{M}_K$. By concatenation of homotopies it suffices to show that any of these curves is homotopic to the unique shortest curve in $\overline{M}_K$ between the endpoints of $\overline{\gamma}_i$ by a homotopy that has the timelike lifting property as above. Since in the case $K=-1,0$ the exponential map is a global diffeomorphism onto $\overline{M}_K$, this homotopy is readily constructed: let $\overline{\gamma} :[0,T] \to \overline{M}_K$ be a piecewise smooth curve with 
\begin{equation}
\label{EqCL1}
||\dot{\overline{\gamma}}(s)||_{\overline{g}_K} <1 \quad \textnormal{ for all } s \in [0,T]\;.
\end{equation}
We define $\OG :[0,T] \times [0,T] \to \overline{M}_K$ by
$$\OG(u,s) = \begin{cases} \exp_{\overline{\gamma}(0)} \Big( \frac{s}{u} \exp^{-1}_{\overline{\gamma}(0)} \big(\overline{\gamma}(u)\big)\Big) \quad &\textnormal{ for } 0 \leq s < u \\
\overline{\gamma}(s) &\textnormal{ for } 0 \leq u \leq s \;. \end{cases} $$
Note that for $0 \leq s <u$ we have $||\rd_s \OG(u,s)||_{\ogk} = \frac{1}{u} ||\exp^{-1}_{\overline{\gamma}(0)} \big(\overline{\gamma}(u)\big) ||_{\ogk} <1$ by \eqref{EqCL1}. Thus, $\OG$ has the timelike lifting property.

In the case $K=+1$ the difficulty is that the curve may pass through the antipodal point of the base point of the exponential map. There, the exponential map is not defined. To circumvent this difficulty we have to use $d \geq 2$ and we perturb the start of the curve slightly so that it passes through a point $p$ such that the curve does not contain $-p$. In a series of steps we show that given a curve $\overline{\gamma} : [0,T] \to \Sp^d$ with $||\dot{\overline{\gamma}}(s)||_{\overline{g}_{\mathbb{S}^d}} <1$ we can find a homotopy with fixed endpoints and with the timelike lifting property that homotopes $\overline{\gamma}$ into either the unique shortest curve from $\overline{\gamma}(0)$ to $\overline{\gamma}(T)$ in the case $\overline{\gamma}(T) \neq -\overline{\gamma}(0)$, or into a fixed geodesic arc from $\overline{\gamma}(0)$ to $\overline{\gamma}(T) = -\overline{\gamma}(0)$.

In the \textbf{first step} we make the curve stationary in the beginning and then accelerate the curve slightly afterwards while preserving that the velocity is bounded away from $1$. Since $\overline{\gamma}$ is piecewise smooth, there exists a $\delta >0$ and an $\varepsilon >0$ such that $||\dot{\overline{\gamma}}(s)||_{\overline{g}_{\mathbb{S}^d}} < 1 - \varepsilon$ for all  $s \in [0,\delta]$. We define $\lambda : [0, \varepsilon\delta] \times [0, T] \to [0,T]$ by
\begin{equation*}
\lambda(u,s) := \begin{cases} 0 &\textnormal{ for } 0 \leq s \leq u \\
(s-u) \frac{\delta}{\delta - u} \quad &\textnormal{ for } u \leq s \leq \delta \\
s &\textnormal{ for }  \delta \leq s \leq T \;,  \end{cases}
\end{equation*}
and set
\begin{equation*}
\overline{\Gamma}_1 (u,s) := \overline{\gamma}\big( \lambda(u,s)\big) \;.
\end{equation*}
Note that $\frac{\delta}{\delta - u} \leq \frac{1}{1-\varepsilon}$ for all $u \in [0, \varepsilon \delta]$, and thus we have $|| \partial_s \overline{\Gamma}_1||_{\ogS} < 1$. It now follows that $\overline{\Gamma}_1 : [0, \varepsilon \delta] \times [0, T] \to \Sp^d$ is a homotopy with fixed endpoints that has the timelike lifting property. We set $\overline{\Gamma}_1(\varepsilon \delta, \cdot) =: \overline{\gamma}^{(1)}(\cdot)$.

In the \textbf{second step} we use the wiggle-room created in the first step to carry out the perturbation. Since $d \geq 2$ we can apply for example Sard's theorem to infer that $\overline{\gamma}([0,T])$ has measure zero in $\Sp^d$. Thus 
\begin{equation}
\label{PfEqPoint}
\parbox{0.86\textwidth}{
\textnormal{there is a point $p \in B_{\frac{\varepsilon \delta}{2}}(\overline{\gamma}(0)) \subseteq \Sp^d$ such that $-p$ is not contained in $\overline{\gamma}([0,T])$ and such that the unique shortest curve from $\overline{\gamma}(0)$ to $p$ does not pass through $-\overline{\gamma}(T)$ except possibly at $\overline{\gamma}(0)$.}}
\end{equation}
We now define $\overline{\Gamma}_2 : [0,1] \times [0,T] \to \Sp^d$ by
\begin{equation*}
\overline{\Gamma}_2(u,s)\\ := \begin{cases} \exp_{\overline{\gamma}(0)}\big[\frac{s}{\nicefrac{\varepsilon\delta}{2}} \cdot u \cdot \exp^{-1}_{\overline{\gamma}(0)}(p)\big] \quad   &\textnormal{ for } 0 \leq s \leq \frac{\varepsilon \delta}{2} \\
\exp_{\overline{\gamma}(0)}\big[\frac{\varepsilon\delta - s}{\nicefrac{\varepsilon\delta}{2}} \cdot u \cdot \exp^{-1}_{\overline{\gamma}(0)}(p)\big] &\textnormal{ for }  \frac{\varepsilon \delta}{2} \leq s \leq \varepsilon \delta \\
\overline{\gamma}^{(1)}(s) &\textnormal{for } \varepsilon \delta \leq s \leq T \;. 
\end{cases}
\end{equation*}
We compute for $s \in [0, \frac{\varepsilon \delta}{2}]$ 
\begin{equation*}
||\partial_s \overline{\Gamma}_2(u,t) ||_{\ogS} = \frac{u}{\nicefrac{\varepsilon \delta}{2}} || \exp^{-1}_{\overline{\gamma}(0)}(p) ||_{\ogS} < u \leq 1  
\end{equation*}
where we used $p \in B_{\frac{\varepsilon \delta}{2}}(\overline{\gamma}(0))$. We set $\overline{\gamma}^{(2)}(\cdot) := \overline{\Gamma}_2(1, \cdot)$. Hence, $\overline{\Gamma}_2$ is a homotopy with fixed endpoints between $\overline{\gamma}^{(1)}$ and $\overline{\gamma}^{(2)}$ that has the timelike lifting property. Note that $\overline{\gamma}^{(2)}(\frac{\varepsilon \delta}{2}) = p$.

In the \textbf{third step} we use the exponential map based at $p$, together with \eqref{PfEqPoint}, to straighten out $\overline{\gamma}^{(2)}|_{[\frac{\varepsilon \delta}{2},T]}$. We define $\overline{\Gamma}_3 : [ \frac{\varepsilon \delta }{2} , T] \times [0, T] \to \Sp^d$ by
\begin{equation*}
\overline{\Gamma}_3(u,s) := \begin{cases} \overline{\gamma}^{(2)}(s) &\textnormal{ for } 0 \leq s \leq \frac{\varepsilon \delta}{2} \\
\exp_{p} \big[\frac{ s- \frac{\varepsilon \delta}{2}}{u - \frac{\varepsilon \delta}{2}} \exp^{-1}_p\big(\overline{\gamma}^{(2)}(u)\big)\big]  \quad &\textnormal{ for } \frac{\varepsilon \delta}{2} \leq s < u \\
\overline{\gamma}^{(2)}(s) &\textnormal{ for } u \leq s \leq T \;.
\end{cases}
\end{equation*}
For $\frac{\varepsilon \delta}{2} \leq s <u$ we compute  $||\rd_s \OG_3(u,s)||_{\ogS} = \frac{1}{u-\frac{\varepsilon \delta}{2}} ||\exp^{-1}_p \big(\overline{\gamma}^{(2)}(u)\big)||_{\ogS} <1$, since $||\dot{\overline{\gamma}}^{(2)}(s)||_{\ogS} <1$.
Setting $\overline{\gamma}^{(3)}(\cdot) := \overline{\Gamma}_3(T,\cdot)$,  it follows that $\overline{\Gamma}_3$ is a homotopy with fixed endpoints between $\overline{\gamma}^{(2)}$ and $\overline{\gamma}^{(3)}$ which has the timelike lifting property.

Finally, in the \textbf{fourth step}, we use the exponential map based at $\overline{\gamma}(T)$ to homotope $\overline{\gamma}^{(3)}$ into a geodesic. We define $\overline{\Gamma}_4 : [0,T) \times [0,T] \to \Sp^d$ by
\begin{equation*}
\overline{\Gamma}_4(u,s):= \begin{cases} \overline{\gamma}^{(3)}(s) &\textnormal{ for } 0 \leq s \leq T-u   \\
\exp_{\overline{\gamma}(T)} \big[ \frac{T-s}{u} \exp^{-1}_{\overline{\gamma}(T)} \big( \overline{\gamma}^{(3)} (T-u)\big)\big] \quad  &\textnormal{ for } T-u <s \leq T \;. \end{cases}
\end{equation*}
This is well-defined by \eqref{PfEqPoint}.
If $\overline{\gamma}(0) \neq - \overline{\gamma}(T)$, then one can extend $\overline{\Gamma}_4$ to $[0,T] \times [0,T]$ by the above definition to yield a homotopy with fixed endpoints between $\overline{\gamma}^{(3)}$ and the unique shortest curve from $\overline{\gamma}(0)$ to $\overline{\gamma}(T)$.  As before one checks that it has the timelike lifting property. 

In the case $\overline{\gamma}(0) = - \overline{\gamma}(T)$, we introduce spherical normal coordinates $(r, \theta)$ for $\Sp^d$ at $\overline{\gamma}(T)$, where $\theta \in \Sp^{d-1}$, and express $\overline{\gamma}^{(3)}(s)$ with respect to these coordinates by $\big(r(s), \theta(s)\big)$. Let $\theta_0 :=\lim_{s \searrow 0} \theta(s)$.  It now follows that $\overline{\Gamma}_4(u,\cdot)$ converges for $u \nearrow T$ to  the geodesic arc parametrised, in the spherical normal coordinates, by $s \mapsto \big(\frac{T-s}{T} \pi, \theta_0\big)$, $s \in (0,T]$. In this case this limit curve furnishes the extension of $\overline{\Gamma}_4$. Finally, consider a fixed geodesic arc from $\overline{\gamma}(0)$ to $\overline{\gamma}(T)$ specified in the above spherical normal coordinates by a point $\theta_1 \in \Sp^{d-1}$. It then remains to choose a continuous curve $\sigma : [0,1] \to \Sp^{d-1}$ with $\sigma(0) = \theta_0$ and $\sigma(1) = \theta_1$ and define $\OG_5 : [0,1] \times [0,T] \to \Sp^d$ in the spherical normal coordinates by $\OG_5(u,s) = \big( \frac{T-s}{T} \pi, \sigma(u)\big)$. Here, we have also crucially used $d \geq 2$.
In either case, concatenation of all homotopies concludes the proof.
\end{proof}

\subsection{The case $K=+1$} \label{SecFLRW1}

\begin{proposition}\label{PropK1}
Let $d \geq 1$ and consider any $(d+1)$-dimensional FLRW spacetime $(M,g)$ as defined in the introduction with $K=+1$. Let $\gamma : (0,1) \to M$ be a future directed timelike curve parameterised by the $t$-coordinate, i.e., $\gamma(s) = (s, \overline{\gamma}(s))$ with $\overline{\gamma}(s) \in \Sp^d$ for $s \in (0,1)$. For any $s_0 \in (0,1)$ we then have
$$\bigcup_{0<s<s_0} I^+\big(\gamma(s),M\big) \cap I^-\big(\gamma(s_0),M\big) = I^-\big(\gamma(s_0),M\big) \;.$$ Furthermore for each such $s_0 \in (0,1)$ there exists a $t_0 >0$ such that  $\{t_0\} \times \Sp^d \subseteq I^-\big(\gamma(s_0),M\big) $.
\end{proposition}

\begin{proof}
We start with a brief observation: consider any timelike curve $\sigma$ in $(M,g)$. We can assume without loss of generality that it is  parameterised by the time coordinate, i.e., $\sigma(s) = (s, \overline{\sigma}(s))$. Since $\sigma$ is a timelike curve we have $0> -1 + a(s)^2\overline{g}_K(\dot{\overline{\sigma}}(s), \dot{\overline{\sigma}}(s))$, i.e.
\begin{equation}
\label{EqCurveTimelike}
\frac{1}{a(s)^2} > \overline{g}_K(\dot{\overline{\sigma}}(s), \dot{\overline{\sigma}}(s))\;.
\end{equation}
Vice versa, any spatial curve $\overline{\sigma}$ mapping into $\Sp^d$ that satisfies \eqref{EqCurveTimelike} gives rise to a timelike curve $\sigma(s) = (s, \overline{\sigma}(s))$.

We first claim that for any $s_0 \in (0,1)$ we have $\bigcup_{0<s<s_0} I^+\big(\gamma(s),M\big) = M$.\footnote{Although not needed for the purposes of this paper, this shows that there is only one TIF, namely $M$ itself.} To see this consider a point $(\hat{t}, \hat{\omega}) \in M = (0,\infty) \times \Sp^d$. Since there are no particle horizons, there exists $0< \tilde{t} < \hat{t}$ such that $\int\limits_{\tilde{t}}^{\hat{t}} \frac{1}{a(t')} \, dt' > \pi$.  Without loss of generality assume $\tilde{t} <s_0$. Since the diameter of $\Sp^d$ is equal to $\pi$, we can now find a spatial curve $\overline{\sigma} : [0,1] \to \Sp^d$ from $\overline{\gamma}(\tilde{t})$ to $\hat{\omega}$ that satisfies \eqref{EqCurveTimelike}. This curve lifts to a future directed timelike curve from $\gamma(\tilde{t})$ to $(\hat{t}, \hat{\omega})$ which shows the claim. This proves the first statement of the proposition. The second statement follows similarly: given $s_0 \in (0,1)$ we choose $0< t_0 < s_0$ such that $\int\limits_{t_0}^{s_0} \frac{1}{a(t')} \, dt' > \pi$. By the above argument any point in $\{t_0\} \times \Sp^d$ can be reached from $\gamma(s_0)$ by a past directed timelike curve.
\end{proof}

Note that this proposition remains true if we replace the spatial geometry of constant positive curvature by any compact spatial geometry.

\subsection{The case $K=-1$} \label{SecFLRW2}

\begin{proposition} \label{PropTIFs}
Let $d \geq 2$ and consider any $(d+1)$-dimensional FLRW spacetime $(M,g)$ as defined in the introduction  with $K=-1$ and without particle horizons. Consider a point $\overline{p} \in \overline{M}_{-1}$ together with spherical normal coordinates centred at $\overline{p}$ for $\overline{M}_{-1}$ as in \eqref{EqMetricNormal}. Set $\tau(t) := \int_1^t \frac{1}{a(t')} \, dt'$ and $u := r + \tau$. 
\begin{enumerate}[a)]
\item Let $\gamma :(0,1] \to M$ be a future directed and past inextendible timelike curve.  We then have either \begin{enumerate}[i)]
\item $\lim_{s \to 0} u \big( \gamma(s) \big) = - \infty$
\end{enumerate}
or
\begin{enumerate}[i)]
\item[ii)] $\lim_{s \to 0} u \big(\gamma(s) \big) = u_0 > - \infty$ and $\lim_{s \to 0} r \big( \gamma(s) \big) = \infty $ and $\lim_{s \to 0} \omega\big(\gamma(s) \big) = \omega_0 \in \Sp^{d-1}$.
\end{enumerate}
\item In the case $i)$ we have $\bigcup_{0 < s} I^+\big( \gamma(s),M\big) = M$.
\item Let $\gamma_i : (0,1] \to M$, $i=1,2$, be two future directed and past inextendible timelike curves with $\lim_{s \to 0} u \big(\gamma_i(s) \big) = u_i > - \infty$ and $\lim_{s \to 0} \omega \big(\gamma_i(s) \big) = \omega_i \in \Sp^{d-1}$. 
\begin{enumerate}[1)]
\item If  $\omega_1 = \omega_2$ and $u_1 = u_2$, then $\bigcup_{s < 0} I^+\big(\gamma_1(s), M\big) = \bigcup_{s < 0}I^+\big(\gamma_2(s), M\big)$.
\item If $\omega_1 = \omega_2$ and $u_1 < u_2$, then $\bigcup_{s < 0} I^+\big(\gamma_1(s), M\big)  \supseteq \bigcup_{s < 0} I^+\big(\gamma_2(s), M\big)$.
\end{enumerate}
\end{enumerate}
\end{proposition}

This proposition begins the characterisation of the TIFs of FLRW spacetimes with $K=-1$ and without particle horizons in terms of $(u, \omega)$-coordinates induced by spherical normal coordinates at a given point $\overline{p} \in \overline{M}_{-1}$. However, the characterisation given here is not complete; with some more work one can show that in the case $c)$, if $(u_1, \omega_1) \neq (u_2, \omega_2)$, then the corresponding TIFs are not the same. However, this is not needed for the purposes of this paper.

\begin{proof}
Note that causal notions only depend on the conformal class of the metric. Using $\tau$ as a coordinate instead of $t$ the metric $g$ takes the form
$$g = a\big(t(\tau) \big)^2 \big[ \underbrace{- d\tau^2 + dr^2 + \sinh^2(r) g_{\Sp^{d-1}}}_{=:h} \big] \;.$$ In the following it will be convenient to work with $(M,h)$ instead of with $(M,g)$.\footnote{Note that the exact Milne spacetime has the same conformal metric $h$. Exact Milne spacetime embeds isometrically into Minkowski spacetime. This embedding can be used to give an alternative proof of Proposition \ref{PropTIFs}.}

\underline{Proof of $a)$:} Since $h(du, \cdot)$ is past directed null, $u \big( \gamma(s) \big)$ is a strictly increasing function of $s$ -- and thus the limit $\lim_{s \to 0} u \big( \gamma(s) \big)$ exists in $\{-\infty\} \cup \R$. 

Assume the case $\lim_{s \to 0} u \big( \gamma(s) \big) = u_0 > - \infty$. We have $u \big(\gamma(s) \big) > u_0$ for all $0<s \leq 1$ and $\tau\big( \gamma(s) \big) \to - \infty$ for $s \to 0$. It thus follows from $r \big(\gamma(s) \big) = u \big(\gamma(s) \big) - \tau \big( \gamma(s) \big)$ that $r \big(\gamma(s) \big) \to \infty$ for $s \to 0$.

To see that $\lim_{s \to 0} \omega \big(\gamma(s) \big) \in \Sp^{d-1}$ exists, parameterise the curve $\gamma$ by $\tau$ and note that\footnote{By abuse of notation we will denote the reparameterised curve again by $\gamma$.} $$0 > h(\dot{\gamma}, \dot{\gamma}) = -1 +\dot{r}^2 + \sinh^2(r) \, ||\dot{\omega}||^2_{\Sp^{d-1}} \;.$$ Moreover, by the above we have $r(\tau) > u_0 - \tau$. Thus, $$||\dot{\omega}||_{\Sp^{d-1}} < \frac{1}{\sinh(r)} < \frac{1}{ \sinh(u_0 - \tau)} \;,$$ which is integrable for $\tau \to - \infty$. This proves $a)$.

\underline{Proof of $b)$:} Let $(\tau, \overline{x}) \in M$. Since there are no particle horizons, there exists $\hat{\tau} < \tau$ such that $(\tau, \overline{x}) \in I^+\big((\hat{\tau}, \overline{p}),M\big)$. Without loss of generality we can assume $\hat{\tau} < u\big(\gamma(1)\big)$. Note that all points on $u = \hat{\tau}$ can be connected to $(\hat{\tau}, \overline{p})$ by a future directed null geodesic.  Since $\lim_{s \to 0} u \big(\gamma(s) \big) = - \infty$, there is an $s_0$ with $u\big(\gamma(s_0) \big) = \hat{\tau}$. The push-up property then gives $(\tau, \overline{x}) \in I^+ \big(\gamma(s_0),M \big)$.

\underline{Proof of $c)$:} To prove $1)$, it suffices, by symmetry, to show $\Ima(\gamma_1) \subseteq \bigcup_{0 < s} I^+ \big(\gamma_2(s),M\big)$. We use $(u,r, \omega)$-coordinates, in which the metric takes the form $h = -du^2 + du \otimes dr + dr \otimes du + \sinh^2(r) g_{\Sp^{d-1}}$. So let $\gamma_1(s_1)$ be given. We show that there is $s_2>0$ such that $\gamma_2(s_2) \in I^- \big(\gamma_1(s_1),M\big)$.

Clearly, we have $\gamma_1 (\frac{s_1}{2}) \in I^- \big(\gamma_1(s_1),M \big)$; and by the openness of the timelike past there is $\lambda >0$ and an open neighbourhood $\id \in V \subseteq SO(d)$ such that $$\big(\gamma_1^u(\frac{s_1}{2}) - \lambda, \gamma_1^u(\frac{s_1}{2} )+ \lambda \big) \times \{\gamma_1^r(\frac{s_1}{2}) \} \times \{ R \gamma_1^\omega(\frac{s_1}{2}) \; | \; R \in V\} \subseteq I^- \big(\gamma_1(s_1),M \big) \;.$$
Since translations in $u$ and rotations $R \in SO(d)$ act as isometries on $(M,h)$, we can act with those on $\gamma_1|_{(0, \frac{s_1}{2}]}$ to infer that
\begin{equation}
\label{EqPartOpenness}
\bigcup\limits_{0 < s \leq \frac{s_1}{2}} \big(\gamma_1^u(s) - \lambda, \gamma_1^u(s)  + \lambda\big) \times \{\gamma_1^r(s) \} \times \{ R \gamma_1^\omega(s) \; | \; R \in V\} \subseteq I^- \big(\gamma_1(s_1),M \big) \;.
\end{equation} 
We claim that $\gamma_2(s_2)$ is contained in the left hand side for $0< s_2$ small enough. To see this, first note that $\gamma_1^\omega(s) \to \omega_1$ and $\gamma^u_1(s) \to u_1$ imply that there is a neighbourhood $\omega_1 \in B \subseteq \Sp^{d-1}$ and a $0< \hat{s}_1$ such that
\begin{equation}
\label{EqPartOp2}
\bigcup\limits_{0 < s < \hat{s}_1} (u_1 - \frac{\lambda}{2}, u_1 + \frac{\lambda}{2}) \times \{\gamma_1^r(s)\} \times B
\end{equation}
is contained in the left hand side of \eqref{EqPartOpenness}. Since $u_1 = u_2$ and $\omega_1 = \omega_2$, we can now choose $0 < s_2$ close enough to $0$ such that $\gamma_2^u(s_2) \in (u_1, u_1 + \frac{\lambda}{2})$, $\gamma_2^\omega(s_2) \in B$, and $\gamma^r_2(s_2) >  \gamma_1^r(\hat{s}_1)$. Since $\gamma_1^r(s) \to \infty$ for $s \to 0$, $\gamma_2(s_2)$ is contained in \eqref{EqPartOp2}. This proves $1)$.

To prove $2)$, we consider, in $(u,r, \omega)$-coordinates, the time-translation $\hat{\gamma}_2(s) := \big(u_2 - u_1 + \gamma_1^u(s), \gamma_1^r(s), \gamma_1^\omega(s) \big)$ of $\gamma_1$. By the time-translation invariance of $h$, $\hat{\gamma}_2$ is also a future directed past inextendible timelike curve with $\lim_{s \to 0} \hat{\gamma}^u_2(s) = u_2$ and $\lim_{s \to 0} \hat{\gamma}^\omega_2(s) = \omega_1 = \omega_2$. We thus infer $\bigcup_{0<s} I^+ \big(\gamma_2(s), M \big) = \bigcup_{0<s} I^+ \big(\hat{\gamma}_2(s),M \big)$ by part $1)$. But now it is obvious that $\hat{\gamma}_2(s) \in I^+ \big(\gamma_1(s), M\big)$ for all $0<s$, from which $2)$ follows.
\end{proof}

The next proposition will be used, in conjunction with the time-reversal of Proposition \ref{PropIntNG}, to relate the causality in a past boundary chart with that of $M$.

\begin{proposition} \label{PropNGI}
Let $d \geq 1$ and consider any $(d+1)$-dimensional FLRW spacetime $(M,g)$ as defined in the introduction  with $K=-1$ and without particle horizons. Consider two points $r,s \in M$ with $r \in I^-(s,M)$. Then no past directed null geodesic emanating from $r$ intersects a past directed null geodesic emanating from $s$.
\end{proposition}

\begin{proof}
Again, this is a statement about the conformal class of $(M,g)$. We set $\tau(t) = \int_1^t \frac{1}{a(t')} \, dt'$ and use it as a coordinate instead of $t$, introduce spherical normal coordinates around $\overline{s} \in \overline{M}_{-1}$, where $s = (\tau_s, \overline{s})$, and work with the conformal metric $h = - d\tau^2 + dr^2 + \sinh^2(r) g_{\Sp^{d-1}}$. The past directed null geodesics from $s$ are then given in $(\tau, r, \omega)$-coordinates by 
\begin{equation}
\label{EqPNG}
[0, \infty) \ni \rho \mapsto (\tau_s - \rho, \rho, \omega_0)
\end{equation}
for some $\omega_0 \in \Sp^{d-1}$ which labels the angular direction of the null geodesic emanating from $s$. If a past directed null geodesic starting from $r \in I^-(s,M)$ intersected one starting from $s$, then, by the push-up property, the point of intersection would lie in the \emph{timelike} past of $s$. However, parameterising any past directed timelike curve starting from $s$ by the $\tau$-coordinate, it is immediate that the spatial radial distance traversed has to be \emph{strictly less} than the $\tau$-coordinate time elapsed. This, however, is a contradiction to \eqref{EqPNG}.
\end{proof}

The following crucial proposition requires the additional assumption from Theorem \ref{Thm2} on the scale factor and spells out the main geometric obstruction to $C^0$-extendibility of these spacetimes which is used in the proof of Theorem \ref{Thm2}.

\begin{proposition} \label{PropBlowUp}
Let $d \geq 2$ and consider any $(d+1)$-dimensional FLRW spacetime $(M,g)$ as defined in the introduction  with $K=-1$ and without particle horizons. Moreover, assume the scale factor satisfies $a(t) \cdot e^{\int_t^1 \frac{1}{a(t')} \, dt'} \to \infty$ for $t \to 0$.  Consider a point $\overline{p} \in \overline{M}_{-1}$ together with spherical normal coordinates centred at $\overline{p}$ for $\overline{M}_{-1}$ as in \eqref{EqMetricNormal}. Set $\tau(t) := \int_1^t \frac{1}{a(t')} \, dt'$, $u := r + \tau$, and define $\Sigma_{t_0} := \{t = t_0\}$. Let $\gamma_i : (0,1] \to M$, $i= 1,2$, be two future directed and past inextendible timelike curves parametrised by the time coordinate $t$ with $\lim_{s \to 0} u \big(\gamma_i(s) \big) > - \infty$ for $i = 1,2$ and $\lim_{s \to 0} \omega \big( \gamma_i(s) \big) = \omega_i  \in \Sp^{d-1}$ with $\omega_1 \neq \omega_2$. Finally, let $q \in M$ be a point with $\Ima(\gamma_i) \subseteq M \setminus J^-(q,M)$ for $i=1,2$.\footnote{The existence of such a point $q$ follows since $\lim_{s \to 0} u \big( \gamma_i(s) \big) > - \infty$.}  We then have
\begin{equation} \label{Late}
d_{\Sigma_{t_0} \setminus J^-(q,M)} \big(\gamma_1(t_0), \gamma_2(t_0)\big) \to \infty \quad \textnormal{ for } t_0 \to 0 \;,
\end{equation}
where the distance in $\Sigma_{t_0} \setminus J^-(q,M)$ is with respect to the induced Riemannian metric on $\Sigma_{t_0} \setminus J^-(q,M)$.
\end{proposition}

\begin{proof}
Since there are no particle horizons, there exists $t_p>0$ such that $p = (t_p, \overline{p}) \in J^-(q,M)$. Thus $J^-(p,M) \subseteq J^-(q,M)$ and $\Sigma_{t_0} \setminus J^-(q,M) \subseteq \Sigma_{t_0} \setminus J^-(p,M)$. It suffices to show that $$d_{\Sigma_{t_0} \setminus J^-(p,M)} \big(\gamma_1(t_0), \gamma_2(t_0)\big) \to \infty \quad \textnormal{ for } t_0 \to 0 \;.$$
In the spherical normal coordinates centred at $\overline{p}$ we have $\Sigma_{t_0} \setminus J^-(p,M) = \{t_0\} \times \{r > \int_{t_0}^{t_p} \frac{1}{a(t')} \, dt'\}$ and the induced metric on this hypersurface is
\begin{equation}\label{EqBoundMetric}
a(t_0)^2\big[ dr^2 + \sinh^2(r) g_{\Sp^{d-1}} \big] \geq a(t_0)^2 \big[ dr^2 + \sinh^2\big( \int_{t_0}^{t_p} \frac{1}{a(t')} \, dt' \big) \cdot g_{\Sp^{d-1}} \big] \;.
\end{equation}
Note that
\begin{equation}\label{EqSinh}
 a(t_0) \sinh\big( \int_{t_0}^{t_p} \frac{1}{a(t')} \, dt' \big) = a(t_0) \frac{e^{\int_{t_0}^{t_p} \frac{1}{a(t')} \,dt'} - e^{- \int_{t_0}^{t_p} \frac{1}{a(t')} \, dt'}}{2} \to \infty \quad \textnormal{ for } t_0 \to 0
\end{equation}
by the assumption on the scale factor. Since $\omega_1 \neq \omega_2$ there exists $\varepsilon>0$ and $\tau_0 >0$ such that 
\begin{equation}
\label{EqSphDist}
d_{\Sp^{d-1}}\big(\omega(\gamma_1(t_0)), \omega(\gamma_2(t_0)) \big) \geq \varepsilon >0 \quad \textnormal{ for all } 0 < t_0 < \tau_0 \;.
\end{equation}
Finally, by the orthogonal form of the metric it is immediate that $$d_{\Sigma_{t_0} \setminus J^-(p,M)} \big(\gamma_1(t_0), \gamma_2(t_0)\big) \geq d_{a(t_0) \sinh \big( \int_{t_0}^{t_p} \frac{1}{a(t')} \, dt' \big) \cdot \Sp^{d-1}} \big( \omega(\gamma_1(t_0)), \omega(\gamma_2(t_0)) \big) \to \infty $$
for $t_0 \to 0$ by \eqref{EqBoundMetric}, \eqref{EqSinh}, and \eqref{EqSphDist}.
\end{proof}

Note that the exclusion of $J^-(q,M)$ in \eqref{Late} is crucial: the distance in $\Sigma_{t_0}$ would go to zero.

\section{Proof of the main theorems}

\subsection{Proof of Theorem \ref{Thm1}: $K=+1$ without particle horizons} \label{SecThm1}

\begin{proof}[Proof of Theorem \ref{Thm1}:]
The proof is by contradiction. Assume there exists a continuous extension $\iota : M \hookrightarrow \tilde{M}$ such that $\partial^-\iota(M) \neq \emptyset$. Let $\tilde{p}$ be a past boundary point. By Proposition \ref{PropBoundaryChart} there exists a past  boundary chart $\tilde{\varphi} : \tilde{U} \to (-\varepsilon_0, \varepsilon_0) \times (-\varepsilon_1, \varepsilon_1)^d$ around $\tilde{p}$. Consider the future directed timelike curve $\tilde{\gamma} : [0, \varepsilon_0) \to \tilde{U}$ given in coordinates by $s \mapsto (s,\underline{0})$. Then $\tilde{\gamma}|_{(0,\varepsilon_0)}$ corresponds to a future directed and past inextendible timelike curve $\gamma := \iota^{-1} \circ \tilde{\gamma}|_{(0,\varepsilon_0)}$ in $M$. It thus follows that we must have $t(\gamma(s)) \to 0$ as $s \to 0$. 

If necessary, we make $\varepsilon_1>0$ smaller to ensure that 
\begin{equation}
\label{PfK1SetUp}
(\frac{19}{20} \varepsilon_0, \varepsilon_0) \times (-\varepsilon_1, \varepsilon_1)^d \subseteq I^+_{\tilde{g}_n}\big((\frac{\varepsilon_0}{2}, \underline{0}), \tilde{U}\big)\;,
\end{equation} 
see also Figure \ref{FigK1}.

We now choose $0<s_0 < \frac{\varepsilon_0}{2}$ sufficiently small such that $J_{\tilde{g}_w}^+\big(0, \tilde{U}\big) \cap J_{\tilde{g}_w}^-\big((s_0, \underline{0}), \tilde{U}\big)$ is compact in $\tilde{U}$.
We claim that the following holds:
\begin{equation}
\label{PfEqualityDiamond}
\iota \Big[ \bigcup_{0<s<s_0} I^+\big(\gamma(s), M\big) \cap I^- \big(\gamma(s_0),M\big) \Big] = I^+\big(\tilde{p}, \tilde{U} \big) \cap I^- \big(\tilde{\gamma}(s_0), \tilde{U}\big) \;.
\end{equation}
To prove this we first observe that since $\tilde{\gamma}$ is a \emph{smooth} timelike curve in $(\tilde{U}, \tilde{g})$, and, moreover, since the timelike future defined with respect to piecewise smooth timelike curves is open, we have $$I^+\big(\tilde{p}, \tilde{U} \big) = \bigcup_{0<s<s_0} I^+\big(\tilde{\gamma}(s), \tilde{U}\big)\;,$$ see also Proposition 2.7 in \cite{Sbie15}. It thus remains to prove $$\iota \big[I^+\big(\gamma(s), M\big) \cap I^- \big(\gamma(s_0),M\big)\big] =  I^+\big(\tilde{\gamma}(s), \tilde{U}\big)\cap I^- \big(\tilde{\gamma}(s_0), \tilde{U}\big)$$
for any $0<s<s_0$. The inclusion $\supseteq$ follows directly from  property \eqref{PropF1} of Proposition \ref{PropBoundaryChart} and \eqref{EqAddPropBC} (taking the time-reversal into account when translating from future boundary charts to past boundary charts). And the inclusion $\subseteq$ follows from Proposition \ref{PropFutureOne} together with Lemma \ref{LemCausalHomotopy}. This proves the claim.

By Proposition \ref{PropK1} it now follows that there exists a $0<t_0$ such that 
\begin{equation}
\label{PfContra}
\iota\big( \{t_0\} \times \Sp^d \big) \subseteq I^+\big(\tilde{p}, \tilde{U} \big) \cap I^- \big(\tilde{\gamma}(s_0), \tilde{U}\big)\;.
\end{equation}
\begin{figure}[h]
\centering
 \def\svgwidth{7cm}
   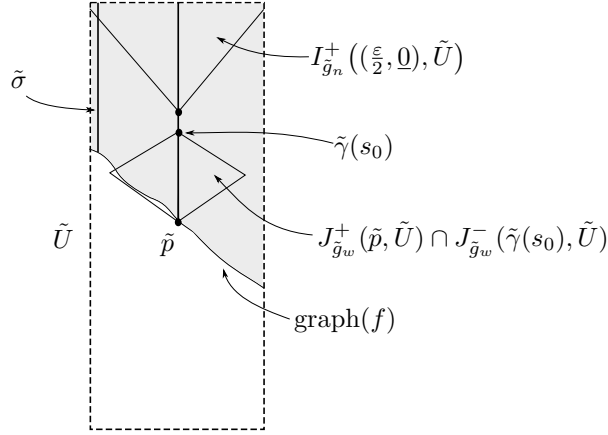
      \caption{The proof of Theorem \ref{Thm1}.} \label{FigK1}
\end{figure}

Since  the set $I^+\big(\tilde{p}, \tilde{U} \big) \cap I^- \big(\tilde{\gamma}(s_0), \tilde{U}\big)$ is compactly contained in $\tilde{U}$, there is an $\underline{x}_0 \in (-\varepsilon_1, \varepsilon_1)^d$ such that the timelike curve $\tilde{\sigma} : \big(f(\underline{x}_0), \varepsilon_0) \to \tilde{U}$, given by $s \mapsto (s, \underline{x}_0)$, does not enter this set. Here, $f$ is the graphing function from Proposition \ref{PropBoundaryChart}. The curve $\sigma:=\iota^{-1} \circ \tilde{\sigma}$ is now a future directed timelike curve in $M$ which is past-inextendible, i.e., we have $t(\sigma(s)) \to 0$ as $s \to f(\underline{x}_0)$. To the future it contains a point in $I^+\big(\gamma(s_0),M\big) \subseteq I^+\big(\{t_0\} \times \Sp^d,M\big)$ by \eqref{PfK1SetUp}, but it does not intersect the Cauchy hypersurface $\{t_0\} \times \Sp^d$ by \eqref{PfContra}. This is a contradiction, which concludes the proof of Theorem \ref{Thm1}.
\end{proof}

Let us remark that another way to prove Theorem \ref{Thm1} is to follow the first steps of the proof of Theorem \ref{Thm2} below to show that there is a point $\tilde{s} \in \tilde{U}_>$ and points $\tilde{p}, \tilde{q} \in \gr(f)$ such that $I^-(\tilde{s}, \tilde{U}) \cap I^+(\tilde{p}, \tilde{U})$ and $I^-(\tilde{s}, \tilde{U}) \cap I^+(\tilde{q}, \tilde{U})$ are non-empty and not the same. Using the future one-connectedness of $(M,g)$ together with Lemma \ref{LemCausalHomotopy} one concludes, as in the above proof, that these sets correspond, via $\iota^{-1}$, to causal diamonds of past directed and past inextendible timelike curves in $M$ starting from $s = \iota^{-1}(\tilde{s})$. This, however, is a contradiction to Proposition \ref{PropK1} by which all such sets have to be the same.

\subsection{Proof of Theorem \ref{Thm2}: $K=-1$ without particle horizons} \label{SecThm2}

\begin{proof}
The proof is by contradiction. Assume there exists a continuous extension $\iota : M \hookrightarrow \tilde{M}$ such that $\partial^-\iota(M) \neq \emptyset$. 

\textbf{Step 1: The set-up.} By the time-reversal of Proposition \ref{PropSpacelikeB} there exists a $\tilde{p} \in \rd^- \iota(M)$ together with a past boundary chart $\tilde{\varphi} : \tilde{U} \to (- \varepsilon_0, \varepsilon_0) \times (-\varepsilon_1, \varepsilon_1)^d$ around $\tilde{p}$ such that the graphing function $f: (-\varepsilon_1, \varepsilon_1)^d \to (-\varepsilon_0, \varepsilon_0)$ is differentiable at $0$ with $\rd_1 f(0) = 0$. If necessary, we make $\varepsilon_1>0$ smaller such that
\begin{equation}
\label{PfFutureCon}
\{ \frac{18}{20} \varepsilon_0 < x_0 < \varepsilon_0 \} \subseteq I^+_{\tilde{g}_n}(\tilde{p}, \tilde{U}) \;.
\end{equation}
It follows from $\rd_1f(0) = 0$ that there exists $0<\mu \leq \varepsilon_1$ such that
\begin{equation}
\label{PfGraphCone}
|f(x_1, 0, \ldots, 0)| < \frac{1}{10} |x_1| \quad \textnormal{ for all } |x_1| < \mu \;.
\end{equation}
Now choose $\varepsilon_0 > \alpha > 0$ such that for $\tilde{s} :=( \alpha, \underline{0})$ and $\tilde{t}:= (-\alpha, \underline{0})$ the set $J^+_{\tilde{g}_w}(\tilde{t}, \tilde{U}) \cap J_{\tilde{g}_w}^-(\tilde{s}, \tilde{U})$ is compact in $\tilde{U}$.
\begin{figure}[h]
\centering
 \def\svgwidth{7cm}
   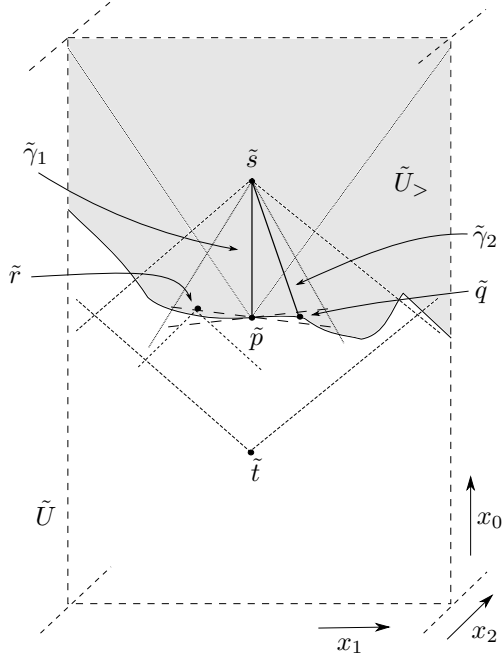
     \caption{The set-up in the proof of Theorem \ref{Thm2} together with lightcone bounds.} \label{FigK-1}
\end{figure}
Moreover choose $0 < \beta < \mu$ so small such that
\begin{equation}
\label{EqQR}
\begin{split}
&\tilde{q}:= \big( f(\beta, 0, \ldots, 0), \beta, 0, \ldots, 0\big) \in I^-_{\tilde{g}_n}(\tilde{s}, \tilde{U}) \cap J^+_{\tilde{g}_w}(\tilde{t}, \tilde{U}) \\
\textnormal{and } \; &\tilde{r}:= \big(\frac{1}{10} \beta, - \beta, 0, \ldots, 0\big) \in I_{\tilde{g}_n}^-(\tilde{s}, \tilde{U}) \;.
\end{split}
\end{equation}
Note that by the bound \eqref{PfGraphCone} we have $\tilde{r} \in \tilde{U}_>$. Basic trigonometry using the lightcone bound $\tilde{g} \prec \tilde{g}_w$ gives that 
\begin{equation}
\label{EqA}
I^+(\tilde{p}, \tilde{U}) \cap I^-(\tilde{s}, \tilde{U}) \; \textnormal{ and } \; I^+(\tilde{q}, \tilde{U}) \cap I^-(\tilde{s}, \tilde{U}) \; \textnormal{ are non-empty and neither set is contained in the other;}
\end{equation}
and
\begin{equation}
\label{EqB}
J^-(\tilde{r}, \tilde{U}_>) \; \textnormal{ does not contain the set } \; \{(x_0, x_1, 0, \ldots, 0) \in \tilde{U} \; | \; 0 \leq x_1 \leq \beta, \; x_0 > f(x_1, 0, \ldots, 0)\} \;.
\end{equation}
We also define the past directed timelike curves $\tilde{\gamma}_i : [0,1] \to \tilde{U}$, $i=1,2$, in the local coordinates given by $\tilde{\varphi}$, by $\tilde{\gamma}_1^\mu(\tau) = \tilde{s}^\mu + \tau (\tilde{p}^\mu - \tilde{s}^\mu)$ and by $\tilde{\gamma}_2^\mu(\tau) = \tilde{s}^\mu + \tau( \tilde{q}^\mu - \tilde{s}^\mu)$. These curves are straight lines in the $(x_0, x_1)$-plane from $\tilde{s}$ to $\tilde{p}$ and $\tilde{q}$, respectively. Note that by \eqref{EqQR} the curve $\tilde{\gamma}_2$ is indeed timelike; and from the achronality of $\gr (f)$ it follows that $\tilde{\gamma}_i|_{[0,1)}$ map into $\tilde{U}_>$. Thus, $\gamma_i := \iota^{-1} \circ \tilde{\gamma}_i|_{[0,1)}$, $i=1,2$, are past directed and past inextendible timelike curves in $M$ starting from $s:= \iota^{-1}(\tilde{s}) \in M$.

\textbf{Step 2: Relating the set-up in $\tilde{U}$ to $M$.} 
As in the proof of Theorem \ref{Thm1}, the future one-connectedness of $(M,g)$, i.e., Proposition \ref{PropFutureOne}, together with Lemma \ref{LemCausalHomotopy} gives
\begin{equation*}
\iota \Big[ \bigcup_{0<\tau<1} I^+\big(\gamma_1(\tau), M\big) \cap I^- \big(s,M\big) \Big] = I^+\big(\tilde{p}, \tilde{U} \big) \cap I^- \big(\tilde{s}, \tilde{U}\big)
\end{equation*}
and 
\begin{equation*}
\iota \Big[ \bigcup_{0<\tau<1} I^+\big(\gamma_2(\tau), M\big) \cap I^- \big(s,M\big) \Big] = I^+\big(\tilde{q}, \tilde{U} \big) \cap I^- \big(\tilde{s}, \tilde{U}\big) \;.
\end{equation*}
It thus follows from \eqref{EqA} that neither $ \bigcup_{0<\tau<1} I^+\big(\gamma_1(\tau), M\big)$ is contained in $ \bigcup_{0<\tau<1} I^+\big(\gamma_2(\tau), M\big) $, nor the other way around.
Let $s = (s_t, \overline{s})$ and introduce spherical normal coordinates around $\overline{s}$ for $\overline{M}_{-1}$ as in \eqref{EqMetricNormal}. We now apply Proposition \ref{PropTIFs}. By part $a)$ $i)$ and $b)$ we cannot have $\lim_{\tau \to 0} u \big(\gamma_i(\tau) \big) = -\infty$ for $i=1$ nor for $i=2$, since otherwise at least one of the TIFs would equal $M$ and thus contain the other. Thus, by  part $a)$ the limits $\lim_{\tau \to 0} u\big(\gamma_i(\tau)\big) = u_i > -\infty$ and $\lim_{\tau \to 0} r \big(\gamma_i(\tau) \big) = \infty$ and $\lim_{\tau \to 0} \omega \big(\gamma_i(\tau) \big) = \omega_i \in \Sp^{d-1}$ exist for $i = 1,2$. By part $c)$ we must have $\omega_1 \neq \omega_2$.

Consider now the point $r := \iota^{-1}(\tilde{r}) \in M$. We claim that
\begin{equation}
\label{EqClai}
\begin{split}
\Ima(\gamma_1) &\cup \Ima(\gamma_2) \subseteq M \setminus J^-(r,M) \\ &\textnormal{ and for some } \varepsilon>0 \\ \iota^{-1} \Big( \{(x_0, x_1, 0, \ldots, 0) \in \tilde{U}_> &\; | \; 0 \leq x_1 \leq \beta, \; 0<x_0 - f(x_1, 0, \ldots, 0) < \varepsilon\} \Big) \subseteq M \setminus J^-(r,M)\;.
\end{split}
\end{equation}
By \eqref{EqQR} we have $r \in I^-(s,M)$ and by Proposition \ref{PropNGI} no past directed null geodesic emanating from $r$ intersects a past directed null geodesic emanating from $s$. We now apply the time-reversal of Proposition \ref{PropIntNG} with our choice of $\tilde{t}$ and $\tilde{s}$ to infer that
\begin{equation}
\label{EqinM}
\iota\big(J^-(r,M) \big) \cap \Big( \big[ J^-(\tilde{s}, \tilde{U}_>) \cap J^+_{\tilde{g}_w}(\tilde{t}, \tilde{U}) \big] \setminus J^-(\tilde{r}, \tilde{U}_>) \Big) = \emptyset \;.
\end{equation}
The basic lightcone bound $\tilde{g} \prec \tilde{g}_w$ gives that  $$\Ima(\tilde{\gamma}_i|_{[0,1)}) \subseteq \big[ J^-(\tilde{s}, \tilde{U}_>) \cap J^+_{\tilde{g}_w}(\tilde{t}, \tilde{U}) \big] \setminus J^-(\tilde{r}, \tilde{U}_>) $$ so that the first part of \eqref{EqClai} follows from \eqref{EqinM}. For the second part we note that there exists $\varepsilon>0$ such that $$\{(x_0, x_1, 0, \ldots, 0) \in \tilde{U}_> \; | \; 0 \leq x_1 \leq \beta, \; 0<x_0 - f(x_1, 0, \ldots, 0) < \varepsilon\}  \subseteq J^-(\tilde{s}, \tilde{U}_>) \cap J^+_{\tilde{g}_w}(\tilde{t}, \tilde{U}) \;.$$ Thus, the second part of \eqref{EqClai} follows from  \eqref{EqB} and \eqref{EqinM}.

The first part of \eqref{EqClai}, together with $\omega_1 \neq \omega_2$ and $u_i > - \infty$ for $i=1,2$, puts us in the setting of Proposition \ref{PropBlowUp} (with $q=r$) and thus we infer that 
\begin{equation}
\label{EqBlowUp}
d_{\Sigma_{t_0} \setminus J^-(r,M)} \big(\gamma_1(t_0), \gamma_2(t_0)\big) \to \infty \quad \textnormal{ for } t_0 \to 0 \;,
\end{equation}
where $\Sigma_{t_0} = \{t=t_0\}$ as in Proposition \ref{PropBlowUp} and where we have reparameterised $\gamma_1$ and $\gamma_2$ by $t$.

\textbf{Step 3: Showing that \eqref{EqBlowUp} is in contradiction with the set-up in $\tilde{U}$.} Consider the point $\tilde{p}' := (\min\{\frac{1}{20} \varepsilon_0, \frac{3}{4} \alpha\}, \underline{0})$, and set $p':= \iota^{-1} (\tilde{p}')$. It then follows from \eqref{PfFutureCon} that
\begin{equation}
\label{PfFutureCon2}
I^+_{\tilde{g}_n}(\tilde{p}', \tilde{U}) \supseteq \{ \frac{19}{20} \varepsilon_0 < x_0 < \varepsilon_0\} \;.
\end{equation}
We now claim that
\begin{equation}
\label{FinalClaim}
\parbox{0.86\textwidth}{\textnormal{there exists $C>0$, depending only on $\delta, \varepsilon_1 >0$, such that for any $t_0>0$ sufficiently small, one can connect $\Ima(\gamma_1) \cap \Sigma_{t_0}$ to $\Ima(\gamma_2) \cap \Sigma_{t_0}$ by a smooth curve $\sigma_{t_0} : [0,1] \to \Sigma_{t_0} \setminus J^-(r,M)$ of length $L(\sigma_{t_0}) \leq C$.}}
\end{equation}
This claim is obviously in contradiction to \eqref{EqBlowUp} for $t_0>0$ small enough. It thus remains to prove \eqref{FinalClaim}.

We first choose $t_0$ small enough such that $p' \in I^+(\Sigma_{t_0}, M)$. By \eqref{PfFutureCon2} in combination with (the time reversal of) \eqref{PropF1} and \eqref{EqAddPropBC} this implies $ I^+(\Sigma_{t_0}, M) \supseteq \iota^{-1}(\{ \frac{19}{20} \varepsilon_0 < x_0 < \varepsilon_0\})$. Thus, for any $\underline{x} \in (-\varepsilon_1, \varepsilon_1)^d$ the curve $\big( f(\ux),\varepsilon_0 \big) \ni s \mapsto (s, \ux)$ is a future directed and past inextendible timelike curve in $M$ which contains points in $I^+(\Sigma_{t_0}, M)$. It thus has to intersect the Cauchy hypersurface $\Sigma_{t_0}$. Hence, there is a function $h_{t_0} : (-\varepsilon_1, \varepsilon_1)^d \to (- \varepsilon_0, \varepsilon_0)$ such that $\iota(\Sigma_{t_0}) \cap \tilde{U}_>$ is the graph of $h_{t_0}$. Since $\iota(\Sigma_{t_0})$ is a smooth hypersurface, $h_{t_0}$ is smooth. Moreover, for $i=1, \ldots, d$, the vector $\rd_i h_{t_0}(\ux) \cdot \rd_0 + \rd_i \in T_{\big(h_{t_0}(\ux), \ux \big)} \tilde{U}$ is clearly tangent to $\iota(\Sigma_{t_0}) \cap \tilde{U}$. Since all tangents of $\iota(\Sigma_{t_0}) \cap \tilde{U}$ have to be spacelike we obtain from the lightcone bound $\tilde{g}_n \prec \tilde{g}$ that $\frac{1}{|\rd_i h_{t_0}(\ux)|} > \frac{1}{\sqrt{2}}$, that is
\begin{equation}
\label{EqBoundH}
|\rd_i h_{t_0}(\ux)| < \sqrt{2} \quad \textnormal{ for all } i \in \{1, \ldots, d\} \textnormal{ and } \ux \in (-\varepsilon_1, \varepsilon_1)^d \;.
\end{equation}
Since $s \in I^+(p',M)$, the past directed timelike curves $\gamma_i$, $i=1,2$, have to intersect $\Sigma_{t_0}$ as well. Let $$\Ima(\tilde{\gamma}_2) \cap  \big(\iota(\Sigma_{t_0}) \cap \tilde{U}_> \big) = \big(h_{t_0}(x_{2, t_0}, 0, \ldots, 0), x_{2, t_0}, 0, \ldots, 0\big) \quad \textnormal{ with } 0 < x_{2, t_0} < \beta < \varepsilon_1 \;.$$
We then define a smooth curve $\tilde{\sigma}_{t_0} : [0,1] \to \iota(\Sigma_{t_0}) \cap \tilde{U}_>$ by $$\tilde{\sigma}_{t_0} (\tau) = \big( h_{t_0}(\tau \cdot x_{2,t_0}, 0, \ldots, 0), \tau \cdot x_{2, t_0}, 0 \ldots, 0 \big) \;.$$ We compute $\dot{\tilde{\sigma}}_{t_0}(\tau) = \rd_1 h_{t_0}(\tau \cdot x_{2, t_0}, 0, \ldots, 0) \cdot x_{2, t_0} \cdot \rd_0 + x_{2, t_0} \cdot \rd_1$. Moreover, we have $|\tilde{g}_{\mu \nu}| \leq 1 + \delta$ for all $\mu, \nu \in \{0, \ldots, d\}$ by Proposition \ref{PropBoundaryChart}. Together with \eqref{EqBoundH}  we thus obtain 
\begin{equation*}
L(\tilde{\sigma}_{t_0}) = \int_0^1 \sqrt{ \tilde{g} \big(\dot{\tilde{\sigma}}_{t_0}(\tau), \dot{\tilde{\sigma}}_{t_0}(\tau)\big)} \, d \tau \leq \sqrt{2 \varepsilon_1^2 \cdot(1 + \delta) + 2 \sqrt{2}  \cdot \varepsilon_1^2(1+ \delta) + \varepsilon_1^2(1+ \delta)}=: C < \infty \;.
\end{equation*}
By construction $\iota^{-1} \circ \tilde{\sigma}_{t_0} =: \sigma_{t_0}$ is a curve mapping into $\Sigma_{t_0}$. It remains to show that it does not map into $J^-(r,M)$.

For each $\ux \in (-\varepsilon_1, \varepsilon_1)^d$ the function $t_0 \mapsto h_{t_0}(\ux)$ is strictly monotonically increasing in $t_0$. Moreover, we have the pointwise limit $\lim_{t_0 \to 0} h_{t_0}(\ux) = f(\ux)$ for all $\ux \in (-\varepsilon_1, \varepsilon_1)^d$, since there cannot be a point in $M$ which lies to the past of all $\Sigma_{t_0}$. By Dini's Theorem, the convergence is indeed uniform. Thus, for given $\varepsilon>0$ we can choose $t_0>0$ small enough such that $$\Ima(\tilde{\sigma}_{t_0}) \subseteq \{(x_0, x_1, 0, \ldots, 0) \in \tilde{U}_> \; | \; 0 \leq x_1 \leq \beta, \; 0<x_0 - f(x_1, 0, \ldots, 0) < \varepsilon\}\;. $$ It now follows from the second part of \eqref{EqClai} that $\sigma_{t_0}$ does not map into $J^-(r,M)$. This concludes the proof of Theorem \ref{Thm2}.
\end{proof}

\appendix

\section{Lipschitz hypersurfaces and tangent cones} \label{SecApp}

For the convenience of the reader some elementary facts regarding Lipschitz hypersurfaces are collated here which are used in the proof of Proposition \ref{PropSpacelikeB}.

\begin{definition}
Let $A \subseteq \R^n$, $a \in \R^n$. The \emph{tangent cone} of $A$ at $a$ is defined by $$Tan(A,a) = \{ v \in \R^n \simeq T_a\R^n \; | \; \textnormal{there exist a sequence } A \ni x_k \to a \textnormal{ such that } v = ||v|| \lim_{k \to \infty} \frac{x_k -a}{||x_k -a||} \} \;.$$
\end{definition}
Clearly, $Tan(A,a)$ is a cone: if $v \in Tan(A,a)$, then so is $r \cdot v \in Tan(A,a)$ for $r >0$.

\begin{lemma} \label{LemChangeCoord}
Let $f : \R^m \to \R^n$ and $A \subseteq \R^m$. If $f$ is differentiable at $a \in A$, $Df|_a$ is injective, $f|_A$ is injective, and $(f|_A)^{-1}$ is continuous at $f(a)$, then $$df|_a\big(Tan(A,a)\big) = Tan \big( f(A), f(a)\big)\;.$$
\end{lemma}

\begin{proof}
\underline{$``\subseteq"$:} Let $v \in Tan(A,a)$ and without loss of generality $v \neq 0$, $||v|| =1$. Let $A \ni x_k \to a$ such that $ \frac{x_k -a}{||x_k - a||} \to v$. Then 
\begin{equation*}
||df|_a(v)|| \cdot  \lim_{k \to \infty} \frac{f(x_k) - f(a)}{||f(x_k) -f(a)||} =  ||df|_a(v)|| \cdot  \lim_{k \to \infty} \frac{||x_k - a||}{||f(x_k) - f(a)||} \cdot \lim_{k \to \infty} \frac{f(x_k) - f(a)}{||x_k - a||} \\
= df|_a(v) \;,
\end{equation*}
where we have used that $||df|_a(v)|| \neq 0$ by the injectivity of $df|_a$.

\underline{$`` \supseteq"$:} Let $w \in Tan \big(f(A), f(a) \big)$, without loss of generality $w \neq 0$, $||w|| = 1$. Then there exists a sequence $f(A) \ni y_k \to f(a)$ such that $\frac{y_k - f(a)}{||y_k - f(a)||} \to w$. Since $f|_A$ is injective and its inverse continuous at $f(a)$, we have  $x_k = (f|_A)^{-1}(y_k) \to a$ for $k \to \infty$. After taking a subsequence if necessary we have $ \lim_{k \to \infty} \frac{x_k - a}{||x_k - a||} =: v \in Tan(A,a)$.  Then, using again that $||df|_a(v)|| \neq 0$, we obtain $$df|_a\big( \frac{v}{||df|_a(v)||}\big) = \lim_{k \to \infty} \frac{f(x_k) - f(a)}{||x_k - a||} \cdot \lim_{k \to \infty} \frac{||x_k - a||}{||f(x_k) - f(a)||} = \lim_{k \to \infty} \frac{y_k - f(a)}{||y_k - f(a)||} = w \;.$$ 
\end{proof}

Given an $n$-dimensional smooth manifold $M$ and $B \subseteq M$ we can now define the tangent cone of $B$ at $b \in B$ by 
\begin{equation*}
\begin{split}
Tan(B,b) := \{ v \in T_bM \; | \; &\textnormal{there exists a smooth chart } \varphi : U \to \R^n \textnormal{ around $b$} \\
&\textnormal{such that $d\varphi (v) \in Tan\big(\varphi(U \cap B), \varphi(b)\big)$}\;.
\end{split}
\end{equation*}
By Lemma \ref{LemChangeCoord} this does not depend on the choice of chart.

\begin{lemma} \label{LemDif}
Let $f : \R^n \to \R$ be a Lipschitz continuous function. Then $f$ is differentiable at $0$ with $df|_0 = A \in \R^n$ if, and only if, $$Tan\big( \mathrm{graph}(f), \underbrace{(f(0),0)}_{\in \R^{n+1}}\big) = \begin{pmatrix}
A_1 & \cdots & A_n \\
& \mathrm{id}_n &
\end{pmatrix} \cdot \R^n
\;.$$
\end{lemma}

\begin{proof}
The direction $``\implies"$ is provided by  applying Lemma \ref{LemChangeCoord} to the graphing function $ \R^n \ni x \mapsto (f(x),x) \in \R^{n+1}$. For the direction $``\impliedby"$ we consider for $y_0 \in \R^n$ with $||y_0||=1$ and $h>0$
\begin{equation} \label{SubSe}
\frac{\big(f(h y_0), h y_0\big) - \big( f(0),0 \big)}{h} = \frac{\big(f(h y_0), h y_0\big) - \big( f(0),0 \big)}{||\big(f(h y_0), h y_0\big) - \big( f(0),0 \big)||} \cdot \underbrace{\frac{||\big(f(h y_0), h y_0\big) - \big( f(0),0 \big)||}{h}}_{\leq L} \;,
\end{equation}
where $L>0$ is related to the Lipschitz constant of $f$. Since the factors on the right hand side are bounded, it follows that there exist convergent subsequences $h_n \to 0$. The left hand side shows that the limit must be of the form $\R^{n+1} \ni Y = Y^0 \rd_0 + y_0^i \rd_i$. The expression on the right makes manifest that $Y \in Tan\big( \mathrm{graph}(f), (f(0),0)\big)$. However, by assumption we must then have $Y^0 = A \cdot y_0$. This shows that all subsequences have the same limit and thus the limit $h \to 0$ in \eqref{SubSe} exists and equals
$$\lim_{h \to 0} \frac{\big(f(h y_0), h y_0\big) - \big( f(0),0 \big)}{h}  = (A \cdot y_0, y_0)\;.$$
The first component of this limit says that $f$ is differentiable at $0$ with derivative equal to $A$.
\end{proof}

\begin{lemma} \label{LemLip}
Let $M$ be a smooth $(d+1)$-dimensional manifold and $B \subseteq M$.  Let $p \in B$ and $\varphi_i : U_i \to \R^{d+1}$, $i = 1,2$ be two smooth charts of $M$ with $p \in U_i$ and such that $\varphi_i (U_i \cap B) = \mathrm{graph}(f_i) $ and $\varphi_i(p) = \big(f_i(0),0\big)$ for some Lipschitz continuous functions $f_i : \R^d \to \R$. Then $f_1$ is differentiable at $0$ if, and only if, $f_2$ is differentiable at $0$.
\end{lemma}

\begin{proof}
Assume $f_1$ is differentiable at $0$. Then $F_1 : \R^d \to \R^{d+1}$, given by $F_1(x) = (f_1(x), x)$, satisfies the assumptions of Lemma \ref{LemChangeCoord} at $0$ and thus we obtain $Tan \big(\mathrm{graph}(f_1), (f_1(0), 0)\big) = dF_1|_0 (\R^d)$.  We now consider the change of charts $\varphi_2 \circ \varphi_1^{-1}$ and apply again Lemma \ref{LemChangeCoord} to obtain that $Tan\big(\mathrm{graph}(f_2) , (f_2(0),0)\big)$ is a $d$-dimensional subspace of $T_{(f_2(0),0)}\R^{d+1}$.

For the sequence of points $y_n := \big(f_2(\frac{1}{n}, 0, \ldots, 0), \frac{1}{n}, 0, \ldots, 0\big) \in \gr(f_2)$ the difference quotient $$\frac{y_n - (f_2(0), 0)}{\nicefrac{1}{n}} = \frac{y_n - (f_2(0), 0)}{||y_n - (f_2(0), 0)||} \cdot \underbrace{ \frac{||f_2(\frac{1}{n}, 0, \ldots, 0), \frac{1}{n}, 0, \ldots, 0) - (f_2(0), 0) ||}{\nicefrac{1}{n}}}_{\leq L}$$ is bounded, where $L >0$ is related to the Lipschitz constant of $f_2$, and thus admits a subsequence $y_{n_k}$ such that $$\frac{y_{n_k} - (f_2(0),0)}{\nicefrac{1}{n}} \to (A_1, 1, 0, \ldots, 0) \in Tan\big(\gr(f_2), (f_2(0),0) \big) \;.$$ Similarly we obtain that there exist $A_2, \ldots, A_d \in \R$ such that $(A_i, e_i) \in Tan\big(\gr(f_2), (f_2(0), 0)\big)$, where $e_i \in \R^d$ is the unit vector in the $i$-th coordinate direction. However, we already know that $Tan \big(\gr(f_2), (f_2(0), 0) \big)$ is a $d$-dimensional subspace and so it follows that $$Tan\big( \mathrm{graph}(f_2), (f_2(0),0)\big) = \begin{pmatrix}
A_1 & \cdots & A_d \\
& \mathrm{id}_d &
\end{pmatrix} \cdot \R^d
\;.$$
It then follows from Lemma \ref{LemDif} that $f_2$ is differentiable at $0$.
\end{proof}

\bibliographystyle{acm}
\bibliography{Bibly}

\begin{thebibliography}{10}

\bibitem{AlGrKuSa19}
{\sc Alexander, S.~B., Graf, M., Kunzinger, M., and S{\"a}mann, C.}
\newblock {Generalized cones as Lorentzian length spaces: Causality, curvature,
  and singularity theorems}.
\newblock {\em arXiv:1909.09575\/} (2019).

\bibitem{CaMo20}
{\sc Cavalletti, F., and Mondino, A.}
\newblock {Optimal transport in Lorentzian synthetic spaces, synthetic timelike
  Ricci curvature lower bounds and applications}.
\newblock {\em arXiv:2004.08934\/} (2020).

\bibitem{ChrusKli12}
{\sc Chru\'sciel, P., and Klinger, P.}
\newblock {The annoying null boundaries}.
\newblock {\em J. Phys. Conf. Ser. 968\/} (2018).

\bibitem{GalLin16}
{\sc Galloway, G., and Ling, E.}
\newblock {Some Remarks on the $C^0$-(in)extendibility of Spacetimes}.
\newblock {\em Ann. Henri Poincar\'e 18}, 10 (2017), 3427--3477.

\bibitem{GalLinSbi17}
{\sc Galloway, G., Ling, E., and Sbierski, J.}
\newblock {Timelike completeness as an obstruction to $C^0$-extensions}.
\newblock {\em Comm. Math. Phys. 359}, 3 (2018), 937--949.

\bibitem{GeKrPe72}
{\sc Geroch, R., Kronheimer, E.~H., and Penrose, R.}
\newblock {Ideal Points in Space-Time}.
\newblock {\em Proceedings of the Royal Society of London. Series A,
  Mathematical and Physical Sciences 327}, 1571 (1972), 545--567.

\bibitem{Gr20}
{\sc Graf, M.}
\newblock {Singularity theorems for $C^1$-Lorentzian metrics}.
\newblock {\em Communications in Mathematical Physics 378}, 2 (2020),
  1417--1450.

\bibitem{GraKuSa19}
{\sc Grant, J., Kunzinger, M., and S\"amann, C.}
\newblock {Inextendibility of spacetimes and Lorentzian length spaces}.
\newblock {\em Ann. Glob. Anal. Geom. 55\/} (2019), 133--147.

\bibitem{KuOhScSt22}
{\sc Kunzinger, M., Ohanyan, A., Schinnerl, B., and Steinbauer, R.}
\newblock {The Hawking--Penrose singularity theorem for $C^1$-Lorentzian
  metrics}.
\newblock {\em Communications in Mathematical Physics 391}, 3 (2022),
  1143--1179.

\bibitem{Ling20}
{\sc Ling, E.}
\newblock {The Big Bang is a Coordinate Singularity for $k = -1$ inflationary
  FLRW Spacetimes}.
\newblock {\em Foundations of Physics 50\/} (2020), 385--428.

\bibitem{MinSuhr19}
{\sc Minguzzi, E., and Suhr, S.}
\newblock {Some regularity results for Lorentz-Finsler spaces}.
\newblock {\em Ann. Glob. Anal. Geom. 56\/} (2019), 597--611.

\bibitem{Sbie18}
{\sc Sbierski, J.}
\newblock {On the proof of the $C^0$-inextendibility of the Schwarzschild
  spacetime}.
\newblock {\em J. Phys. Conf. Ser. 968\/} (2018).

\bibitem{Sbie15}
{\sc Sbierski, J.}
\newblock {The $C^0$-inextendibility of the Schwarzschild spacetime and the
  spacelike diameter in Lorentzian geometry}.
\newblock {\em J. Diff. Geom. 108}, 2 (2018), 319--378.

\bibitem{Sbie22a}
{\sc Sbierski, J.}
\newblock {On holonomy singularities in general relativity and the
  ${C_{\mathrm{loc}}^{0,1}}$-inextendibility of space-times}.
\newblock {\em Duke Mathematical Journal 171}, 14 (2022), 2881 -- 2942.

\bibitem{Sbie22b}
{\sc Sbierski, J.}
\newblock {Uniqueness and non-uniqueness results for spacetime extensions}.
\newblock {\em arXiv:2208.07752v1\/} (2022).

\end{thebibliography}

\end{document}